\documentclass[12pt]{article}
\usepackage{lineno,hyperref}
\usepackage{amssymb}
\usepackage{amsfonts}
\usepackage{amsmath}
\usepackage{anysize}

\modulolinenumbers[5]
\marginsize{2cm}{2cm}{1cm}{1cm}
\newtheorem{theorem}{Theorem}[section]

\newtheorem{corollary}[theorem]{Corollary}

\newtheorem{lemma}[theorem]{Lemma}
\newtheorem{notation}[theorem]{Notation}

\newtheorem{proposition}[theorem]{Proposition}
\newtheorem{remark}[theorem]{Remark}

\newenvironment{proof}[1][Proof]{\noindent\textbf{#1.} }{\ \rule{0.5em}{0.5em}}
\input{tcilatex}
\begin{document}

\title{Quantum Fluctuations and Large Deviation Principle for Microscopic
Currents of Free Fermions in Disordered Media}
\author{J.-B. Bru \and W. de Siqueira Pedra \and A. Ratsimanetrimanana}
\date{\today }
\maketitle

\begin{abstract}
We contribute an extension of large-deviation results obtained in [N.J.B.
Aza, J.-B. Bru, W. de Siqueira Pedra, A. Ratsimanetrimanana, J. Math. Pures
Appl. 125 (2019) 209] on conductivity theory at atomic scale of free lattice
fermions in disordered media. Disorder is modeled by (i) a random external
potential, like in the celebrated Anderson model, and (ii) a
nearest-neighbor hopping term with random complex-valued amplitudes. In
accordance with experimental observations, via the large deviation
formalism, our previous paper showed in this case that quantum uncertainty
of microscopic electric current densities around their (classical)
macroscopic value is suppressed, exponentially fast with respect to the
volume of the region of the lattice where an external electric field is
applied. Here, the quantum fluctuations of linear response currents are
shown to exist in the thermodynamic limit and we mathematically prove that
they are related to the rate function of the large deviation principle
associated with current densities. We also demonstrate that, in general,
they do not vanish (in the thermodynamic limit) and the quantum uncertainty
around the macroscopic current density disappears exponentially fast with an
exponential rate proportional to the squared deviation of the current from
its macroscopic value and the inverse current fluctuation, with respect to
growing space (volume) scales.\bigskip

\end{abstract}

%TCIMACRO{\TeXButton{\tableofcontents }{\tableofcontents}}%
%BeginExpansion
\tableofcontents%
%EndExpansion

\section{Introduction}

Surprisingly \cite{Ohm-exp2}, in 2012, experimental measurements \cite%
{Ohm-exp} of electric resistance of nanowires in Si doped with phosphorus
atoms demonstrate that the macroscopic laws for charge transport are already
accurate at length scales larger than a few nanometers, even at very low
temperature ($4.2~\mathrm{K}$). As a consequence, microscopic (quantum)
effects on charge transport can very rapidly disappear with respect to
growing space scales. Understanding the breakdown of the classical
(macroscopic) conductivity theory at microscopic scales is an important
technological issue, because of the growing need for smaller electronic
components.

From a mathematical perspective, the convergence of the expectations of
microscopic current densities with respect to growing space scales is proven
in \cite{OhmIII,OhmVI}, but no information about the suppression of quantum
uncertainty was obtained in the macroscopic limit. In \cite{LDP}, in
accordance with experimental observations, it is proven, for non-interacting
lattice fermions with disorder, that quantum uncertainty of microscopic
electric current densities around their (classical) macroscopic value is
suppressed, exponentially fast with respect to the volume of the region of
the lattice where an external electric field is applied. This is proven in 
\cite{LDP} via the large deviation formalism \cite{DS89,dembo1998large},
which has been adopted in quantum statistical mechanics since the eighties 
\cite[Section 7]{ABPM1}. Given a fixed electromagnetic field $\mathcal{E}$,
we derive in particular in \cite{LDP} the (good) rate function $\mathrm{I}^{(%
\mathcal{E)}}$ associated with microscopic (linear response) current
densities\footnote{%
In some direction of $\mathbb{R}^{d}$.} $x_{L}^{(\mathcal{E)}}\in \mathbb{R}$%
, $L\in \mathbb{R}_{0}^{+}$, meaning in this case that, in a cubic box of
volume $L^{d}$ ($d$-dimmensional lattice), for any $a,b\in \mathbb{R}$, 
\begin{equation}
\mathrm{Prob}\left[ x_{L}^{(\mathcal{E)}}\in \left[ a,b\right] \right] \sim 
\mathrm{e}^{-L^{d}\inf_{x\in \left[ a,b\right] }\mathrm{I}^{(\mathcal{E)}%
}(x)}\ ,\qquad \text{as }L\rightarrow \infty \ ,  \label{ddddddddddddd}
\end{equation}%
with $\mathrm{I}^{(\mathcal{E)}}\geq 0$ and $\mathrm{I}^{(\mathcal{E)}}(x)=0$
iff $x$ is the macroscopic (linear response) current density, $x^{(\mathcal{E%
})}$.

In this paper, we complement these studies by rigorously showing two new
properties of charge transport of quasi-free fermions in disordered media:\ 

\begin{itemize}
\item[(a)] The quantum fluctuations of linear response currents exist in the
thermodynamic limit and are meanwhile explicitly related to the rate
function $\mathrm{I}^{(\mathcal{E)}}$, as expected.

\item[(b)] In general, the quantum fluctuations of currents do not vanish in
the thermodynamic limit and the quantum uncertainty around the macroscopic
current density disappears exponentially fast with an exponential rate
proportional to $(x-x^{(\mathcal{E})})^{2}$ and the inverse current
fluctuation, with respect to growing space (volume) scales.
\end{itemize}

\noindent (a)-(b) refer to Theorems \ref{Quantum fluctuations and rate
function} and \ref{new theorem}, which are the main results of this paper.

Our results show that the experimental measure of the rate function $\mathrm{%
I}^{(\mathcal{E)}}$ (see (\ref{ddddddddddddd})) leads to an experimental
estimate on the corresponding quantum fluctuations. Conversely, an
experimental estimate on these quantum fluctuations gives the behavior of
the corresponding rate function $\mathrm{I}^{(\mathcal{E)}}$ around the
macroscopic current density $x^{(\mathcal{E})}$. This fact is certainly not
restricted to fermionic currents.

Note that the existence of quantum fluctuations and associated mathematical
structures has been extensively studied for quantum many-body systems. This
refers for instance to the construction of so-called algebra of normal
fluctuations for transport phenomena, which are related to quantum central
limit theorems.\ See, e.g., \cite{OhmIII,OhmIV,GVV1,GVV2,GVV3,GVV4,GVV5,GVV6}
as well as \cite[Chapter 6]{Verbeure} and references therein. The explicit
relation (a) we derive between quantum fluctuations and the large deviation
formalism in quantum statistical mechanics \cite[Section 7]{ABPM1} is,
however, a new general observation on quantum many-body systems.

We use the mathematical framework of \cite{LDP,OhmVI,brupedraLR} to study
fermions on the lattice. For simplicity we take a cubic lattice $\mathbb{Z}%
^{d}$, even if other types of lattices can be considered with very similar
methods. Disorder within the conductive material, due to impurities, crystal
lattice defects, etc., is modeled by (i) a random external potential, like
in the celebrated Anderson model, and (ii) a nearest-neighbor hopping term
with random complex-valued amplitudes. In particular, random
(electromagnetic) vector potentials can also be implemented. The celebrated
tight-binding Anderson model is one particular example of the general case
considered here.

In order to prove Property (a), i.e., Theorem \ref{Quantum fluctuations and
rate function}, we use the large deviation formalism and follow the argument
lines of \cite[Section 4]{LDP} to show \cite[Theorem 3.1]{LDP} via the
Akcoglu-Krengel ergodic theorem \cite[Theorem 4.17]{LDP}, for one has to
control the thermodynamical limit of (finite-volume) generating functions
that are random. We perform in particular the same box decomposition of
these random functions, which can be justified with the help of the
Bogoliubov-type inequality \cite[Lemma 4.2]{LDP} and the \textquotedblleft
locality\textquotedblright\ (or space decay) of both the quasi-free dynamics
and space correlations of KMS states, which is a consequence of
Combes-Thomas estimates \cite[Appendix A]{LDP}. See \cite[Section 4.3]{LDP}.
In this paper we only give the new arguments that are necessary to prove
Property (a), like the existence of the thermodynamic limit of quantum
fluctuations of currents and the continuity of the second derivative of the
generating function. In particular, like in the proof of \cite[Corollary 4.20%
]{LDP}, we use the (Arzel\`{a}-) Ascoli theorem \cite[Theorem A5]{Rudin},
which requires uniform bounds on the third-order derivatives of
finite-volume generating functions. This proof is much more computational
than the one of \cite[Proposition 4.9]{LDP}, which only control the first
and second derivatives of the same function. Note that derivatives\ of the
logarithm of the expectations of an exponential, like the generating
function we consider here, are generally related to so-called
\textquotedblleft truncated\textquotedblright\ or \textquotedblleft
connected\textquotedblright\ correlations. We demonstrate that it is the
case for the third-order derivative we refer to above, allowing the reader
to follow the computation of that derivative in a systematic way.
Considering the third-order case, the algorithm to compute the derivatives
of the generating functions at any order becomes apparent, showing that the
generating function is in fact \emph{smooth}. We give below further remarks
on that.

In order to prove Property (b), i.e., Theorem \ref{new theorem} (Theorem \ref%
{Quantum fluctuations and rate function} being proven), we rewrite the
second derivative of the generating function, which is the thermodynamic
limit of the quantum fluctuations of currents (Theorem \ref{Quantum
fluctuations and rate function} (i)), as a trace of some explicit positive
operator in the one-particle Hilbert space. This quantity can be estimated
from below by the Hilbert-Schmidt norm of a kind of current observable in
the one-particle Hilbert space. Various computations and estimates then
imply Theorem \ref{new theorem}.

As discussed in \cite{LDP}, observe the existence of a large mathematical
literature on charged transport properties of fermions in disordered media,
see for instance \cite%
{[SBB98],BvESB94,jfa,klm,JMP-autre,JMP-autre2,germinet,Pro13,Cornean} and
references therein. However, it is not the purpose of this introduction to
go into the details of the history of this specific research field. For a
(non-exhaustive) historical perspective on linear conductivity (Ohm's law),
see, e.g., \cite{brupedrahistoire} or our previous papers \cite%
{OhmIII,OhmVI,OhmIV,brupedraLR,OhmI,OhmII,OhmV}.

To conclude, this paper is organized as follows:

\begin{itemize}
\item In Section \ref{sec:SetPro}, we describe the mathematical framework,
which is the one of \cite{LDP,OhmVI,brupedraLR}. It refers to quasi-free
fermions on the lattice in disordered media. Although all the problem can be
formulated, in a mathematically equivalent way, in the one-particle (or
Hilbert space) setting \cite[Appendix C.3]{LDP}, since the underlying
physical system is a many-body one, it is conceptually more appropriate to
state our results within the algebraic formulation for lattice fermion
systems, like in \cite{LDP,OhmVI,brupedraLR}. Short complementary
discussions on response of quasi-free fermion systems to\ electric fields
can be found in \cite[Appendix C]{LDP}.

\item In Section \ref{sec:main}, the main results are stated. In particular,
Property (a) described above refers to Section \ref{rate_quantum_fluctuation}%
, while Property (b) is explained in Section \ref{Non-Vanishing}.

\item Section \ref{sec:proofs} gathers all technical proofs. In particular,\
Sections \ref{Sectino tech3 copy(2)}-\ref{Sectino tech3} give preliminary
definitions and observations, while Sections \ref{Sectino tech3 copy(1)} and %
\ref{positivite} refer to the proofs of Theorems \ref{Quantum fluctuations
and rate function} (i) and \ref{new theorem}, respectively.
\end{itemize}

\begin{notation}
\label{remark constant}\mbox{
}\newline
A norm on a generic vector space $\mathcal{X}$ is denoted by $\Vert \cdot
\Vert _{\mathcal{X}}$. The Banach space of all bounded linear operators on $(%
\mathcal{X},\Vert \cdot \Vert _{\mathcal{X}}\mathcal{)}$ is denoted by $%
\mathcal{B}(\mathcal{X})$. The scalar product of any Hilbert space $\mathcal{%
X}$ is denoted by $\langle \cdot ,\cdot \rangle _{\mathcal{X}}$. We use the
convention $\mathbb{R}^{+}\doteq \left\{ x\in \mathbb{R}:x>0\right\} $\
while $\mathbb{R}_{0}^{+}\doteq \mathbb{R}^{+}\cup \{0\}$. For any random
variable $X$, $\mathbb{E}[X]$ denotes its expectation and $\mathrm{Var}[X]$
its variance.
\end{notation}

\section{Setup of the Problem\label{sec:SetPro}}

We use the mathematical framework of \cite{LDP,OhmVI,brupedraLR} in order to
study fermions on the lattice.

\subsection{Random Tight-Binding Model\label{Section impurities}}

We consider conducting fermions in a cubic crystal represented by the $d$%
-dimensional cubic lattice $\mathbb{Z}^{d}$ ($d\in \mathbb{N}$). The
corresponding one-particle Hilbert space is thus $\mathfrak{h}\doteq \ell
^{2}(\mathbb{Z}^{d};\mathbb{C})$. Its canonical orthonormal basis is denoted
by $\left\{ \mathfrak{e}_{x}\right\} _{x\in \mathbb{Z}^{d}}$, where $%
\mathfrak{e}_{x}(y)\doteq \delta _{x,y}$ for all $x,y\in \mathbb{Z}^{d}$. ($%
\delta _{x,y}$ is the Kronecker delta.)

Disorder in the crystal is modeled via a probability space $(\Omega ,%
\mathfrak{A}_{\Omega },\mathfrak{a}_{\Omega })$, defined as follows: Using
the sets%
\begin{equation*}
\mathbb{D}\doteq \{z\in \mathbb{C}\colon \left\vert z\right\vert \leq 1\}%
\text{\quad and\quad }\mathfrak{b}\doteq \left\{ \{x,x^{\prime }\}\subseteq 
\mathbb{Z}^{d}\colon |x-x^{\prime }|=1\right\}
\end{equation*}%
we define 
\begin{equation*}
\Omega \doteq \lbrack -1,1]^{\mathbb{Z}^{d}}\times \mathbb{D}^{\mathfrak{b}%
}\qquad \text{and}\qquad \mathfrak{A}_{\Omega }\doteq \left( \otimes _{x\in 
\mathbb{Z}^{d}}\mathfrak{A}_{x}^{(1)}\right) \otimes \left( \otimes _{%
\mathbf{x}\in \mathfrak{b}}\mathfrak{A}_{\mathbf{x}}^{(2)}\right) \ ,
\end{equation*}%
where $\mathfrak{A}_{x}^{(1)}$, $x\in \mathbb{Z}^{d}$, and $\mathfrak{A}_{%
\mathbf{x}}^{(2)}$, $\mathbf{x}\in \mathfrak{b}$, are the Borel $\sigma $%
-algebras of respectively the interval $[-1,1]$ and the unit disc $\mathbb{D}
$, both with respect to their usual metric topology. The distribution $%
\mathfrak{a}_{\Omega }$ is an \emph{ergodic} probability measure on the
measurable space $(\Omega ,\mathfrak{A}_{\Omega })$. See \cite{LDP} for more
details. Below, $\mathbb{E}\left[ \cdot \right] $ and $\mathrm{Var}[\cdot ]$
always refer to expectations and variances associated with $\mathfrak{a}%
_{\Omega }$.

Given $\vartheta \in \mathbb{R}_{0}^{+}$ and $\omega =(\omega _{1},\omega
_{2})\in \Omega $ we define a bounded self-adjoint operator $\Delta _{\omega
,\vartheta }\in \mathcal{B}(\mathfrak{h})$ encoding the hopping amplitudes
of a single particle in the lattice: 
\begin{eqnarray}
\lbrack \Delta _{\omega ,\vartheta }(\psi )](x) &\doteq &2d\psi
(x)-\sum_{j=1}^{d}\Big((1+\vartheta \overline{\omega _{2}(\{x,x-e_{j}\})})\
\psi (x-e_{j})  \notag \\
&&+\psi (x+e_{j})(1+\vartheta \omega _{2}(\{x,x+e_{j}\}))\Big)
\label{equation sup}
\end{eqnarray}%
for any $x\in \mathbb{Z}^{d}$ and $\psi \in \mathfrak{h}$, where $%
\{e_{k}\}_{k=1}^{d}$ is the canonical basis of $\mathbb{R}^{d}$. If $%
\vartheta =0$, $\Delta _{\omega ,0}$ is (up to a minus sign) the usual $d$%
-dimensional discrete Laplacian. Random (electromagnetic) vector potentials
can also be implemented in our model, since $\omega _{2}$ takes values in
the unit disc $\mathbb{D}\subseteq \mathbb{C}$. Then, the random
tight-binding model is the one-particle Hamiltonian defined by 
\begin{equation}
h^{(\omega )}\doteq \Delta _{\omega ,\vartheta }+\lambda \omega _{1}\ ,\text{%
\qquad }\omega =\left( \omega _{1},\omega _{2}\right) \in \Omega ,\ \lambda
,\vartheta \in \mathbb{R}_{0}^{+},  \label{eq:Ham_lap_pot}
\end{equation}%
where the function $\omega _{1}\colon \mathbb{Z}^{d}\rightarrow \lbrack
-1,1] $ is identified with the corresponding (self-adjoint) multiplication
operator. The celebrated tight-binding Anderson model corresponds to the
special case $\vartheta =0$.

\subsection{$C^{\ast }$-Algebraic Setting\label{subsec:algebraic}}

We denote by $\mathcal{U}$ the universal unital $C^{\ast }$-algebra
generated by elements $\{a(\psi )\}_{\psi \in \mathfrak{h}}$ satisfying the
canonical anticommutation relations (CAR): For all $\psi ,\varphi \in 
\mathfrak{h}$, 
\begin{equation}
a(\psi )a(\varphi )=-a(\varphi )a(\psi ),\quad a(\psi )a(\varphi )^{\ast
}+a(\varphi )^{\ast }a(\psi )=\left\langle \psi ,\varphi \right\rangle _{%
\mathfrak{h}}\mathfrak{1}.  \label{CAR}
\end{equation}%
As is usual, $a(\psi )$ and $a(\psi )^{\ast }$ refer to, respectively,
annihilation and creation operators in the fermionic Fock space
representation.

For all $\omega \in \Omega $ and $\lambda ,\vartheta \in \mathbb{R}_{0}^{+}$%
, a dynamics on the $C^{\ast }$-algebra $\mathcal{U}$ is defined by the
unique strongly continuous group $\tau ^{(\omega )}\doteq (\tau
_{t}^{(\omega )})_{t\in {\mathbb{R}}}$ of (Bogoliubov) $\ast $-automorphisms
of $\mathcal{U}$ satisfying%
\begin{equation}
\tau _{t}^{(\omega )}(a(\psi ))=a(\mathrm{e}^{ith^{(\omega )}}\psi )\ ,\text{%
\qquad }t\in \mathbb{R},\ \psi \in \mathfrak{h}.  \label{rescaledbis}
\end{equation}%
See (\ref{eq:Ham_lap_pot}) as well as \cite[Theorem 5.2.5]%
{bratteli2003operator2} for more details on Bogoliubov automorphisms.

For any realization $\omega \in \Omega $ and disorder strengths $\lambda
,\vartheta \in \mathbb{R}_{0}^{+}$, the thermal equilibrium state of the
system at inverse temperature $\beta \in \mathbb{R}^{+}$ (i.e., $\beta >0$)
is by definition the unique $(\tau ^{(\omega )},\beta )$-KMS state $\varrho
^{(\omega )}$, see \cite[Example 5.3.2.]{bratteli2003operator2} or \cite[%
Theorem 5.9]{AttalJoyePillet2006a}. It is well-known that such a state is
stationary with respect to the dynamics $\tau ^{(\omega )}$, that is, 
\begin{equation*}
\varrho ^{(\omega )}\circ \tau _{t}^{(\omega )}=\varrho ^{(\omega )}\
,\qquad \omega \in \Omega ,\ t\in \mathbb{R}.
\end{equation*}%
The state $\varrho ^{(\omega )}$ is also gauge-invariant, quasi-free and
satisfies%
\begin{equation}
\varrho ^{(\omega )}(a^{\ast }\left( \varphi \right) a\left( \psi \right)
)=\left\langle \psi ,\frac{1}{1+\mathrm{e}^{\beta h^{(\omega )}}}\varphi
\right\rangle _{\mathfrak{h}},\qquad \varphi ,\psi \in \mathfrak{h}.
\label{2-point correlation function}
\end{equation}%
The gauge-invariant quasi-free state with two-point correlation functions
given by (\ref{2-point correlation function}) for $\beta =0$ is the tracial
state (or chaotic state), denoted by $\text{tr}\in \mathcal{U}^{\ast }$.

Recall that gauge-invariant quasi-free states are positive linear
functionals $\rho \in \mathcal{U}^{\ast }$ such that $\rho (\mathfrak{1})=1$
and, for all $N_{1},N_{2}\in \mathbb{N}$ and $\psi _{1},\ldots ,\psi
_{N_{1}+N_{2}}\in \mathfrak{h}$, 
\begin{equation}
\rho \left( a^{\ast }(\psi _{1})\cdots a^{\ast }(\psi _{N_{1}})a(\psi
_{N_{1}+N_{2}})\cdots a(\psi _{N_{1}+1})\right) =0  \label{ass O0-00}
\end{equation}%
if $N_{1}\neq N_{2}$, while in the case $N_{1}=N_{2}\equiv N$, 
\begin{equation}
\rho \left( a^{\ast }(\psi _{1})\cdots a^{\ast }(\psi _{N})a(\psi
_{2N})\cdots a(\psi _{N+1})\right) =\mathrm{det}\left[ \rho \left( a^{\ast
}(\psi _{k})a(\psi _{N+l})\right) \right] _{k,l=1}^{N}.  \label{ass O0-00bis}
\end{equation}%
See, e.g., \cite[Definition 3.1]{Araki}, which refers to a more general
notion of quasi-free states. The gauge-invariant property corresponds to
Equation (\ref{ass O0-00}) whereas \cite[Definition 3.1, Condition (3.1)]%
{Araki} only imposes the quasi-free state to be even, which is a strictly
weaker property than being gauge-invariant.

\subsection{Linear Response Current Density\label{Current Densities}}

\noindent \underline{(i) Paramagnetic currents:} Fix $\omega \in \Omega $
and $\vartheta \in \mathbb{R}_{0}^{+}$. For any oriented edge $(x,y)\in (%
\mathbb{Z}^{d})^{2}$, we define the paramagnetic\footnote{%
Diamagnetic currents correspond to the ballistic movement of charged
particles driven by electric fields. Their presence leads to the progressive
appearance of paramagnetic currents which are responsible for heat
production. For more details, see \cite{OhmVI,OhmII,OhmV} as well as \cite[%
Appendix C]{LDP} on linear response currents.} current observable by 
\begin{equation}
I_{(x,y)}^{(\omega )}\doteq -2\Im \mathrm{m}\left( \langle \mathfrak{e}%
_{x},\Delta _{\omega ,\vartheta }\mathfrak{e}_{y}\rangle _{\mathfrak{h}}a(%
\mathfrak{e}_{x})^{\ast }a(\mathfrak{e}_{y})\right) ,
\label{current observable}
\end{equation}%
where, as is usual, the real and imaginary parts of any element $A\in 
\mathcal{U}$ are respectively defined by 
\begin{equation}
\Re \mathrm{e}\left( A\right) \doteq \frac{1}{2}\left( A+A^{\ast }\right)
\quad \text{and}\quad \Im \mathrm{m}\left( A\right) \doteq \frac{1}{2i}%
\left( A-A^{\ast }\right) .  \label{im and real part}
\end{equation}%
The self-adjoint elements $I_{(x,y)}^{(\omega )}\in \mathcal{U}$\ are seen
as current observables because they satisfy a discrete continuity equation,
as explained in \cite[Appendix C]{LDP}. This \textquotedblleft
second-quantized\textquotedblright\ definition of a current observable and
the usual one in the one-particle setting, like in \cite{[SBB98],jfa,klm},
are perfectly equivalent, in the case of non-interacting fermions. See for
instance \cite[Appendix C.3]{LDP}. \medskip

\noindent \underline{(ii) Conductivity:} As is usual, $\left[ A,B\right]
\doteq AB-BA\in \mathcal{U}$ denotes the commutator between the elements $%
A\in \mathcal{U}$ and $B\in \mathcal{U}$. For any finite subset $\Lambda
\subsetneq \mathbb{Z}^{d}$, we define the space-averaged transport
coefficient observable $\mathcal{C}_{\Lambda }^{(\omega )}\in C^{1}(\mathbb{R%
};\mathcal{B}(\mathbb{R}^{d};\mathcal{U}^{d}))$, with respect to the
canonical basis $\{e_{q}\}_{q=1}^{d}$ of $\mathbb{R}^{d}$, by the
corresponding matrix entries%
\begin{eqnarray}
\left\{ \mathcal{C}_{\Lambda }^{(\omega )}\left( t\right) \right\} _{k,q}
&\doteq &\frac{1}{\left\vert \Lambda \right\vert }\underset{%
x,y,x+e_{k},y+e_{q}\in \Lambda }{\sum }\int\nolimits_{0}^{t}i[\tau _{-\alpha
}^{(\omega )}(I_{\left( y+e_{q},y\right) }^{(\omega )}),I_{\left(
x+e_{k},x\right) }^{(\omega )}]\mathrm{d}\alpha  \notag \\
&&+\frac{2\delta _{k,q}}{\left\vert \Lambda \right\vert }\underset{x\in
\Lambda }{\sum }\Re \mathrm{e}\left( \langle \mathfrak{e}_{x+e_{k}},\Delta
_{\omega ,\vartheta }\mathfrak{e}_{x}\rangle a(\mathfrak{e}_{x+e_{k}})^{\ast
}a(\mathfrak{e}_{x})\right)  \label{defininion para coeff observable}
\end{eqnarray}%
for any $\omega \in \Omega $, $t\in \mathbb{R}$, $\lambda ,\vartheta \in 
\mathbb{R}_{0}^{+}$ and $k,q\in \{1,\ldots ,d\}$. It is the conductivity
observable matrix associated with the lattice region $\Lambda $ and time $t$%
. See \cite[Appendix C]{LDP}. In fact, the first term in the right-hand side
of (\ref{defininion para coeff observable}) corresponds to the paramagnetic
coefficient, whereas the second one is the diamagnetic component. For more
details, see \cite[Theorem 3.7]{OhmV}. \medskip

\noindent \underline{(iii) Linear response current density:} Fix a direction 
$\vec{w}\in {\mathbb{R}}^{d}$ with $\left\Vert \vec{w}\right\Vert _{\mathbb{R%
}^{d}}=1$ and a (time-dependent) continuous, compactly supported, electric
field $\mathcal{E}\in C_{0}^{0}(\mathbb{R};\mathbb{R}^{d})$, i.e., the
external electric field is a continuous function $t\mapsto \mathcal{E}(t)\in 
\mathbb{R}^{d}$ of time $t\in \mathbb{R}$, with compact support. Then, as it
is explained in \cite[Appendix C]{LDP} as well as in \cite{OhmVI,OhmV}%
\footnote{%
Strictly speaking, these papers use smooth electric fields, but the
extension to the continuous case is straightforward.}, the space-averaged
linear response current observable in the lattice region $\Lambda $ and at
time $t=0$ in the direction $\vec{w}$ is equal to%
\begin{equation}
\mathbb{I}_{\Lambda }^{(\omega ,\mathcal{E})}\doteq \underset{k,q=1}{\sum^{d}%
}w_{k}\int_{-\infty }^{0}\left\{ \mathcal{E}\left( \alpha \right) \right\}
_{q}\left\{ \mathcal{C}_{\Lambda }^{(\omega )}\left( -\alpha \right)
\right\} _{k,q}\mathrm{d}\alpha .  \label{current}
\end{equation}%
By \cite{OhmIII, OhmVI}, the macroscopic (linear response) current density
produced by electric fields $\mathcal{E}\in C_{0}^{0}(\mathbb{R};\mathbb{R}%
^{d})$ at time $t=0$ in the direction $\vec{w}$ is consequently equal to 
\begin{equation}
x^{(\mathcal{E})}\doteq \lim_{L\rightarrow \infty }\mathbb{E}\left[ \varrho
^{(\cdot )}\left( \mathbb{I}_{\Lambda _{L}}^{(\cdot ,\mathcal{E})}\right) %
\right] \in \mathbb{R},  \label{def currents0}
\end{equation}%
where $\Lambda _{L}\doteq \{\mathbb{Z}\cap \left[ -L,L\right] \}^{d}$ for
any $L\in \mathbb{R}_{0}^{+}$. In order to obtain the current density at any
time $t\in \mathbb{R}$ in the direction $\vec{w}$, it suffices to replace $%
\mathcal{E}\in C_{0}^{0}(\mathbb{R};\mathbb{R}^{d})$ in the last two
equations with 
\begin{equation}
\mathcal{E}_{t}(\alpha )\doteq \mathcal{E}\left( \alpha +t\right) ,\qquad
\alpha \in \mathbb{R}.  \label{new field}
\end{equation}%
For a short summary on response of quasi-free fermion systems to\ electric
fields, see \cite[Appendix C]{LDP}.

\subsection{Large Deviations for Microscopic Current Densities}

Fix again a direction $\vec{w}\in {\mathbb{R}}^{d}$ with $\left\Vert \vec{w}%
\right\Vert _{\mathbb{R}^{d}}=1$ and a time-dependent electric field $%
\mathcal{E}\in C_{0}^{0}(\mathbb{R};\mathbb{R}^{d})$. Recall that $\Lambda
_{L}\doteq \{\mathbb{Z}\cap \left[ -L,L\right] \}^{d}$ for any $L\in \mathbb{%
R}_{0}^{+}$. From \cite{OhmIII, OhmVI} combined with \cite[Corollary 3.2]%
{LDP}, it follows that the distributions\footnote{%
Here, like in \cite{LDP}, the distribution associated to a selfadjoint
element $A$ of a unital $C^{\ast }$-algebra $\mathfrak{A}$ and to a state on
this algebra is the probability measure on the spectrum of $A$ representing
the restriction of the state to the unital $C^{\ast }$-subalgebra of\ $%
\mathfrak{A}$\ generated by $A$. Recall that this measure exists and is
unique, by the Riesz-Markov representation theorem.} of the microscopic
current density observables $(\mathbb{I}_{\Lambda _{L}}^{(\omega ,\mathcal{E}%
)})_{L\in \mathbb{R}^{+}}$, in the state $\varrho ^{(\omega )}$, weak$^{\ast
}$ converge, for $\omega \in \Omega $ almost surely, to the delta
distribution at the macroscopic value $x^{(\mathcal{E})}$, well-defined by
Equation (\ref{def currents0}). By \cite[Corollary 3.5]{LDP}, the quantum
uncertainty around the macroscopic value disappears \emph{exponentially fast}%
, as $L\rightarrow \infty $.

To arrive at that conclusion we use in \cite{LDP} the large deviation
formalism for the microscopic (linear response) current density in the state 
$\varrho ^{(\omega )}$. More precisely, we prove in \cite[Corollary 3.2]{LDP}
that, almost surely\footnote{%
The measurable subset $\tilde{\Omega}\subseteq \Omega $ of full measure of 
\cite[Corollary 3.2]{LDP} does not depend on $\beta \in \mathbb{R}^{+}$, $%
\vartheta ,\lambda \in \mathbb{R}_{0}^{+}$, $\mathcal{E}\in C_{0}^{0}(%
\mathbb{R};\mathbb{R}^{d})$ and $\vec{w}\in {\mathbb{R}}^{d}$ with $%
\left\Vert \vec{w}\right\Vert _{\mathbb{R}^{d}}=1$.} (or with probability
one in $\Omega $), for any borel subset $\mathcal{G}$ of $\mathbb{R}$ with
interior and closure respectively denoted by $\mathcal{G}^{\circ }$ and $%
\mathcal{\bar{G}}$, 
\begin{multline*}
-\inf_{x\in \mathcal{G}^{\circ }}\mathrm{I}^{(\mathcal{E)}}\left( x\right)
\leq \liminf_{L\rightarrow \infty }\frac{1}{\left\vert \Lambda
_{L}\right\vert }\ln \varrho ^{(\omega )}\left( \mathbf{1}\left[ \mathbb{I}%
_{\Lambda _{L}}^{(\omega ,\mathcal{E})}\in \mathcal{G}\right] \right) \\
\leq \limsup_{L\rightarrow \infty }\frac{1}{\left\vert \Lambda
_{L}\right\vert }\ln \varrho ^{(\omega )}\left( \mathbf{1}\left[ \mathbb{I}%
_{\Lambda _{L}}^{(\omega ,\mathcal{E})}\in \mathcal{G}\right] \right) \leq
-\inf_{x\in \mathcal{\bar{G}}}\mathrm{I}^{(\mathcal{E)}}\left( x\right) .
\end{multline*}%
By an abuse of notation\footnote{%
In fact, the object $\varrho ^{(\omega )}\left( \mathbf{1}\left[ \mathbb{I}%
_{\Lambda _{L}}^{(\omega ,\mathcal{E})}\in \mathcal{G}\right] \right) $ can
be easily given a precise mathematical sense by using the (up to unitary
equivalence) unique cyclic representation of the $C^{\ast }$-algebra $%
\mathcal{U}$ associated to the state $\varrho ^{(\omega )}$, noting that the
bicommutant of a $\ast $-algebra\ in any representation is a von Neumann
algebra and thus admits a mesurable calculus.}, we applied above the
(non-continuous) characteristic function $\mathbf{1}\left[ x\in \mathcal{G}%
\right] $ to $\mathbb{I}_{\Lambda _{L}}^{(\omega ,\mathcal{E})}$. Here, by 
\cite[Theorems 3.1, 3.4, Corollary 3.2]{LDP}, the so-called good\footnote{%
It means, in this context, that $\{x\in \mathbb{R}\colon \mathrm{I}^{(%
\mathcal{E)}}(x)\leq m\}$ is compact for any $m\geq 0$.} rate function $%
\mathrm{I}^{(\mathcal{E)}}$ is a deterministic, positive,
lower-semicontinuous, convex function defined by 
\begin{equation}
\mathrm{I}^{(\mathcal{E)}}(x)\doteq \sup\limits_{s\in \mathbb{R}}\left\{ sx-%
\mathrm{J}^{(s\mathcal{E})}\right\} \geq 0,\qquad x\in {\mathbb{R}},
\label{rate_function}
\end{equation}%
where%
\begin{equation}
\mathrm{J}^{(\mathcal{E})}\doteq \lim_{L\rightarrow \infty }\frac{1}{%
\left\vert \Lambda _{L}\right\vert }\mathbb{E}\left[ \ln \varrho ^{(\cdot
)}\left( \mathrm{e}^{\left\vert \Lambda _{L}\right\vert \mathbb{I}_{\Lambda
_{L}}^{(\cdot ,\mathcal{E})}}\right) \right] \in {\mathbb{R}}
\label{generating_function}
\end{equation}%
for all $\beta \in \mathbb{R}^{+}$, $\vartheta ,\lambda \in \mathbb{R}%
_{0}^{+}$, $\mathcal{E}\in C_{0}^{0}(\mathbb{R};\mathbb{R}^{d})$ and $\vec{w}%
\in {\mathbb{R}}^{d}$ with $\left\Vert \vec{w}\right\Vert _{\mathbb{R}%
^{d}}=1 $. By \cite[Theorem 3.4]{LDP}, $\mathrm{I}^{(\mathcal{E)}}$
restricted to the interior of its domain is continuous and, as clearly
expected, the rate function $\mathrm{I}^{(\mathcal{E)}}$ vanishes on the
macroscopic (linear response) current density $x^{(\mathcal{E})}$, i.e., $%
\mathrm{I}^{(\mathcal{E)}}(x^{(\mathcal{E})})=0$, whereas $\mathrm{I}^{(%
\mathcal{E)}}(x)>0$ for all $x\neq x^{(\mathcal{E})}$.

For any $\mathcal{E}\in C_{0}^{0}(\mathbb{R};\mathbb{R}^{d})$, note that
Equation (\ref{rate_function}) means that $\mathrm{I}^{(\mathcal{E)}}$ is
the Legendre-Fenchel transform of the generating function $s\mapsto \mathrm{J%
}^{(s\mathcal{E})}$ from ${\mathbb{R}}$ to itself, which is a well-defined,
continuously differentiable, convex function, by \cite[Theorem 3.1]{LDP}.
Moreover, by \cite[Corollary 4.20 and Equation (54)]{LDP}, for any $\beta
\in \mathbb{R}^{+}$, $\vartheta ,\lambda \in \mathbb{R}_{0}^{+}$, $\mathcal{E%
}\in C_{0}^{0}(\mathbb{R};\mathbb{R}^{d})$, $\vec{w}\in {\mathbb{R}}^{d}$
with $\left\Vert \vec{w}\right\Vert _{\mathbb{R}^{d}}=1$, the macroscopic
current density defined by (\ref{def currents0}) can be expressed in terms
of the generating function: 
\begin{equation}
x^{(\mathcal{E})}=\partial _{s}\mathrm{J}^{(s\mathcal{E})}|_{s=0}\ .
\label{def currents}
\end{equation}

\section{Main Results\label{sec:main}}

In order to provide a rather complete study of conductivity at the atomic
scale for free-fermions in a lattice, we analyse here the rate function
defined by Equation (\ref{rate_function}) in much more detail than in \cite%
{LDP}. See \cite[Corollary 3.2]{LDP}. We focus on the behavior of the rate
function near the macroscopic value of the current density (see (\ref{def
currents})), because it establishes a very interesting connection between
exponential suppression of quantum uncertainties at the atomic scale and the
concept of \emph{quantum fluctuations}, in the case of currents.

\subsection{Quantum Fluctuations of Linear Response Currents and Rate
Function\label{rate_quantum_fluctuation}}

For any inverse temperature $\beta \in \mathbb{R}^{+}$, disorder strengths $%
\vartheta ,\lambda \in \mathbb{R}_{0}^{+}$, disorder realization $\omega \in
\Omega $, direction $\vec{w}\in {\mathbb{R}}^{d}$ with $\left\Vert \vec{w}%
\right\Vert _{\mathbb{R}^{d}}=1$ and time-dependent electric field $\mathcal{%
E}\in C_{0}^{0}(\mathbb{R};\mathbb{R}^{d})$, the quantum fluctuations of
linear response currents in cubic boxes are defined to be%
\begin{equation}
\mathbf{F}_{L}^{(\omega ,\mathcal{E})}\doteq |\Lambda _{L}|\left( \varrho
^{(\omega )}\left( \left( \mathbb{I}_{\Lambda _{L}}^{(\omega ,\mathcal{E}%
)}\right) ^{2}\right) -\varrho ^{(\omega )}\left( \mathbb{I}_{\Lambda
_{L}}^{(\omega ,\mathcal{E})}\right) ^{2}\right) \geq 0,\qquad L\in \mathbb{R%
}_{0}^{+},  \label{qq}
\end{equation}%
with $\Lambda _{L}\doteq \{\mathbb{Z}\cap \left[ -L,L\right] \}^{d}$ and $%
\mathbb{I}_{\Lambda _{L}}^{(\omega ,\mathcal{E})}(t)$ being the
space-averaged linear response current defined by (\ref{current}). Observe
that 
\begin{equation*}
\left\vert \Lambda _{L}\right\vert \varrho ^{(\omega )}(\mathbb{I}_{\Lambda
_{L}}^{(\omega ,\mathcal{E})}(t)),\qquad L\in \mathbb{R}_{0}^{+},
\end{equation*}%
is the (total) current linear response (in the direction $\vec{w}$) to the
electric field and, consequently, 
\begin{equation}
\mathbf{F}_{L}^{(\omega ,\mathcal{E})}=\frac{1}{|\Lambda _{L}|}\left(
\varrho ^{(\omega )}\left( \left( \left\vert \Lambda _{L}\right\vert \mathbb{%
I}_{\Lambda _{L}}^{(\omega ,\mathcal{E})}\right) ^{2}\right) -\varrho
^{(\omega )}\left( \left\vert \Lambda _{L}\right\vert \mathbb{I}_{\Lambda
_{L}}^{(\omega ,\mathcal{E})}\right) ^{2}\right) ,\qquad L\in \mathbb{R}%
_{0}^{+},  \label{qq2}
\end{equation}%
are naturally seen as (normal) quantum fluctuations of the (total) linear
response current. Note that these quantum fluctuations are not quite the
same current fluctuations of \cite{OhmIII,OhmIV}, which correspond only to
the paramagnetic component of the current, whereas ($\mathbf{F}_{L}^{(\omega
,\mathcal{E})}$) also includes the diamagnetic one and thus refers to the
total current.

Recall that $x^{(\mathcal{E})}$ is the macroscopic (linear response) current
density defined by (\ref{def currents0}) and $\mathrm{I}^{(\mathcal{E)}}$ (%
\ref{rate_function}) is the (good) rate function associated with the large
deviation principle of the sequence $\{\mathbb{I}_{\Lambda _{L}}^{(\omega ,%
\mathcal{E})}\}_{L\in \mathbb{R}^{+}}$ of microscopic current densities, in
the KMS state $\varrho ^{(\omega )}$ and with speed $\left\vert \Lambda
_{L}\right\vert $. See, e.g., \cite[Theorems 3.1, 3.4, Corollary 3.2]{LDP}.
We are now in a position to connect the quantum fluctuations of (linear
response) currents with the generating and rate functions associated with
the large deviation principle of microscopic current densities.

\begin{theorem}[Quantum fluctuations and rate function]
\label{Quantum fluctuations and rate function}\mbox{}\newline
There is a measurable subset $\tilde{\Omega}\subseteq \Omega $ of full
measure such that, for all $\beta \in \mathbb{R}^{+}$, $\vartheta ,\lambda
\in \mathbb{R}_{0}^{+}$, $\omega \in \tilde{\Omega}$, $\mathcal{E}\in
C_{0}^{0}(\mathbb{R};\mathbb{R}^{d})$ and $\vec{w}\in {\ \mathbb{R}}^{d}$
with $\left\Vert \vec{w}\right\Vert _{\mathbb{R}^{d}}=1$, the following
properties hold true: \newline
\emph{(i)} The generating function $s\mapsto \mathrm{J}^{(s\mathcal{E})}$
defined by (\ref{generating_function}) belongs to $C^{\infty }\left( \mathbb{%
R};\mathbb{R}\right) $ and satisfies 
\begin{equation}
\partial _{s}^{2}\mathrm{J}^{(s\mathcal{E})}|_{s=0}=\lim_{L\rightarrow
\infty }\mathbb{E}\left[ \mathbf{F}_{L}^{(\cdot ,\mathcal{E})}\right]
=\lim_{L\rightarrow \infty }\mathbf{F}_{L}^{(\omega ,\mathcal{E})}\geq 0.%
\text{ }  \label{sd}
\end{equation}%
\emph{(ii)} The rate function $\mathrm{I}^{(\mathcal{E)}}$ satisfies the
asymptotics 
\begin{equation*}
\mathrm{I}^{(\mathcal{E)}}(x)=\frac{1}{2\partial _{s}^{2}\mathrm{J}^{(s%
\mathcal{E})}|_{s=0}}\left( x-x^{(\mathcal{E})}\right) ^{2}+o\left( \left(
x-x^{(\mathcal{E})}\right) ^{2}\right) ,
\end{equation*}%
provided that $\partial _{s}^{2}\mathrm{J}^{(s\mathcal{E})}|_{s=0}\neq 0$.
\end{theorem}

\begin{proof}
Fix all parameters of the theorem. By Corollary \ref{final corrollary}, the
generating function $s\mapsto \mathrm{J}^{(s\mathcal{E})}$ belongs to $%
C^{2}\left( \mathbb{R};\mathbb{R}\right) $ and satisfies (\ref{sd}). As
explained after Corollary \ref{final corrollary}, under the assumptions of
Theorem \ref{Quantum fluctuations and rate function}, one can
straightforwardly extend our arguments to prove that the generating function 
$s\mapsto \mathrm{J}^{(s\mathcal{E})}$ defined by (\ref{generating_function}%
) is infinitely differentiable. (i) thus holds true. It remains to prove
Assertion (ii): Since the map $s\mapsto \mathrm{J}^{(s\mathcal{E})}$ from ${%
\mathbb{R}}$ to itself is convex and belongs (at least) to $C^{1}\left( 
\mathbb{R};\mathbb{R}\right) $ (see, e.g., Assertion (i) or \cite[Theorem 3.1%
]{LDP}), all finite solutions $s(x)\in \mathbb{R}$ of the variational
problem (\ref{rate_function}) for $x\in \mathbb{R}$, i.e., 
\begin{equation}
\mathrm{I}^{(\mathcal{E)}}(x)=s(x)x-\mathrm{J}^{(s(x)\mathcal{E})},
\label{definition rate}
\end{equation}%
satisfy 
\begin{equation}
x=f(s(x)),  \label{blabla_1}
\end{equation}%
with $f$ being the real-valued function defined by 
\begin{equation}
f(s)\doteq \partial _{s}\mathrm{J}^{(s\mathcal{E})},\qquad s\in \mathbb{R}.
\label{def f}
\end{equation}%
Assume now that $\partial _{s}^{2}\mathrm{J}^{(s\mathcal{E})}|_{s=0}\neq 0$,
which is equivalent in this case to%
\begin{equation}
\partial _{s}f(0)=\partial _{s}^{2}\mathrm{J}^{(s\mathcal{E})}|_{s=0}>0,
\label{blabla_2}
\end{equation}%
by positivity of fluctuations (see (i)). Since, by Corollary \ref{final
corrollary}, the mapping $s\mapsto \mathrm{J}^{(s\mathcal{E})}$ from ${%
\mathbb{R}}$ to itself belongs (at least) to $C^{2}\left( \mathbb{R};\mathbb{%
R}\right) $, by the inverse function theorem combined with (\ref{definition
rate})-(\ref{blabla_2}) and (\ref{def currents}), there is an open interval 
\begin{equation*}
\mathcal{I}\subseteq \left\{ f(s):s\in {\mathbb{R}}\text{ such that }%
\partial _{s}f(s)>0\right\} \subseteq {\mathbb{R}}
\end{equation*}%
containing $x^{(\mathcal{E})}=f(0)$ and a $C^{1}$-function $x\mapsto s(x)$
from $\mathcal{I}$ to ${\mathbb{R}}$ such that Equations (\ref{definition
rate})-(\ref{def f}) hold true. In particular, 
\begin{equation}
\partial _{s}f(s(x))=\partial _{s}^{2}\mathrm{J}^{(s\mathcal{E}%
)}|_{s=s(x)}>0,\qquad x\in \mathcal{I}.  \label{postive2}
\end{equation}%
Clearly,%
\begin{equation}
\partial _{x}s(x)=\frac{1}{\partial _{s}f\left( s(x)\right) },\qquad x\in 
\mathcal{I}.  \label{def f2}
\end{equation}%
We thus infer from (\ref{definition rate})-(\ref{def f}) and (\ref{def f2}),
together with (i), that 
\begin{equation*}
\partial _{x}\mathrm{I}^{(\mathcal{E)}}(x)=s(x),\qquad x\in \mathcal{I}.
\end{equation*}%
Consequently, $\partial _{x}\mathrm{I}^{(\mathcal{E)}}$ is differentiable on 
$\mathcal{I}$ with derivative given by 
\begin{equation*}
\partial _{x}^{2}\mathrm{I}^{(\mathcal{E)}}(x)=\partial _{x}s(x),\qquad x\in 
\mathcal{I}.
\end{equation*}%
As a consequence, $\mathrm{I}^{(\mathcal{E)}}$ is twice differentiable on $%
\mathcal{I}\supseteq \{x^{(\mathcal{E})}\}$ and, using the Taylor theorem at
the point $x^{(\mathcal{E})}$, one obtains that 
\begin{equation}
\mathrm{I}^{(\mathcal{E)}}(x)=s(x^{(\mathcal{E})})\left( x-x^{(\mathcal{E}%
)}\right) +\frac{1}{2}\partial _{x}s(x^{(\mathcal{E})})\left( x-x^{(\mathcal{%
E})}\right) ^{2}+o\left( \left( x-x^{(\mathcal{E})}\right) ^{2}\right) ,
\label{blabla_0}
\end{equation}%
provided (\ref{blabla_2}) holds true. Since, by (\ref{def currents}), (\ref%
{def f}) and (\ref{def f2}), $s(x^{(\mathcal{E})})=0$ and 
\begin{equation*}
\partial _{x}s(x^{(\mathcal{E})})=\frac{1}{\partial _{s}f\left( 0\right) }%
=\left. \frac{1}{\partial _{s}^{2}\mathrm{J}^{(s\mathcal{E})}}\right\vert
_{s=0},
\end{equation*}%
one thus deduces (ii) from (\ref{blabla_0}).
\end{proof}

This theorem is a very interesting observation on the physics of fermionic
systems because it shows that the experimental measure of the rate function
of currents around the expected value leads to an experimental estimate on
the corresponding quantum fluctuations. Conversely, by Theorem \ref{Quantum
fluctuations and rate function}, an experimental estimate on these quantum
fluctuations gives the behavior of the corresponding rate function around
the expected value. This phenomenon is certainly not restricted to fermionic
currents and this is a new observation on transport properties of quantum
many-body systems, to our knowledge.

\begin{remark}[Extension of Theorem \protect\ref{Quantum fluctuations and
rate function}]
\mbox{}\newline
The proof of Theorem \ref{Quantum fluctuations and rate function} can be
generalized to very general kinetic terms (i.e., it does not really depend
on the special choice $\Delta _{\omega ,\vartheta }$), provided the pivotal
Combes-Thomas estimate holds true for the one-particle Hamiltonian. Note,
however, that this would require a new, more complicated, definition of
currents, which results from the commutator of the density operator at fixed
lattice site with the kinetic term (cf. continuity equations on the CAR
algebra \cite[Eqs. (38)--(39)]{OhmV}). We did not implement this
generalization here, because we think that, conceptually, the gain is too
small as compared to the drawbacks concerning notations, definitions, and
technical proofs. Instead, we aim at obtaining an extension of Theorem \ref%
{Quantum fluctuations and rate function} to weakly interacting fermionic
systems by using new constructive methods based on Grassmann-Berezin
integrals, Brydges-Kennedy expansions, etc.
\end{remark}

\subsection{Non-Vanishing Quantum Fluctuations of Linear Response Currents 
\label{Non-Vanishing}}

By Theorem \ref{Quantum fluctuations and rate function}, the behavior of the
rate function within a neighborhood of the macroscopic current densities is
directly related to the quantum fluctuations of the linear response current,
provided these fluctuations do not vanish in the thermodynamic limit, i.e.,
if $\partial _{s}^{2}\mathrm{J}^{(s\mathcal{E})}|_{s=0}\neq 0$ (see Theorem %
\ref{Quantum fluctuations and rate function} (i)). We do not expect this
situation to appear in presence of disorder. We discuss this issue in
Section \ref{positivite}, where we give sufficient conditions ensuring
non-vanishing quantum fluctuations of linear response currents in the
thermodynamic limit. This study leads to the following theorem:

\begin{theorem}[Sufficient conditions for non-zero quantum fluctuations]
\label{new theorem}\mbox{}\newline
Take $\vartheta ,\lambda \in \mathbb{R}_{0}^{+}$, $T,\beta \in \mathbb{R}%
^{+} $, $\mathcal{E}\in C_{0}^{0}(\mathbb{R};\mathbb{R}^{d})$ with support
in $[-T,0]$ and $\vec{w}\doteq (w_{1},\ldots ,w_{d})\in {\mathbb{R}}^{d}$
with $\left\Vert \vec{w}\right\Vert _{\mathbb{R}^{d}}=1$. Assume that the
random variables $\{\omega _{1}\left( z\right) \}_{z\in \mathbb{Z}^{d}}$ are
independently and identically distributed (i.i.d.). Then, for sufficiently
small $T$ and $\vartheta $,%
\begin{equation*}
\partial _{s}^{2}\mathrm{J}^{(s\mathcal{E})}|_{s=0}\geq \frac{\lambda
^{2}\Upsilon ^{(\mathcal{E},\vec{w})}}{\left( 1+\mathrm{e}^{\beta \left(
2d\left( 2+\vartheta \right) +\lambda \right) }\right) ^{2}}\mathrm{Var}%
\left[ (\cdot )_{1}\left( 0\right) \right]
\end{equation*}%
with 
\begin{equation*}
\Upsilon ^{(\mathcal{E},\vec{w})}\mathbf{\doteq }\left( \int_{-\infty
}^{0}\left\langle w,\mathcal{E}\left( \alpha \right) \right\rangle _{\mathbb{%
R}^{d}}\alpha ^{2}\mathrm{d}\alpha \right) ^{2}+\frac{1}{2}%
\sum_{k=1}^{d}\left( w_{k}\int_{-\infty }^{0}\left( \mathcal{E}\left( \alpha
\right) \right) _{k}\alpha ^{2}\mathrm{d}\alpha \right) ^{2}.
\end{equation*}%
In particular, $\partial _{s}^{2}\mathrm{J}^{(s\mathcal{E})}|_{s=0}\neq 0$
whenever $\Upsilon ^{(\mathcal{E},\vec{w})}>0$, $\omega _{1}\left( 0\right) $
is not almost surely constant (and thus $\mathrm{Var}\left[ (\cdot
)_{1}\left( 0\right) \right] >0$, by Chebychev's inequality) and $%
T,\vartheta $ are sufficiently small.
\end{theorem}

\begin{proof}
This is a direct consequence of Equations (\ref{inequalty fluctuation 1})
and (\ref{iid}), in Section \ref{sec:proofs}.
\end{proof}

By Theorems \ref{Quantum fluctuations and rate function} and \ref{new
theorem}, we thus demonstrate that, in general, the quantum fluctuations of
linear response currents do not vanish in the thermodynamic limit and the
quantum uncertainty around the macroscopic current density $x^{(\mathcal{E}%
)} $ disappears exponentially fast, as the volume of the cubic box $\Lambda
_{L} $ grows, with a rate proportional to the squared deviation of the
current from $x^{(\mathcal{E})}$ and the inverse current fluctuation. In
particular, by combining Theorem \ref{Quantum fluctuations and rate function}
(i) with Theorem \ref{new theorem} we can obtain an \emph{explicit} upper
bound on the rate function $\mathrm{I}^{(\mathcal{E)}}$ around $x^{(\mathcal{%
E})}$.

The fact that the random variables $\{\omega _{1}\left( z\right) \}_{z\in 
\mathbb{Z}^{d}}$ are independently and identically distributed (i.i.d.) in
Theorem \ref{new theorem} is not essential here: For any $\omega \in \Omega $%
, let $w^{(\omega )}\doteq (w_{1}^{(\omega )},\ldots ,w_{d}^{(\omega )})\in 
\mathbb{R}^{d}$ be the random vector defined by 
\begin{equation*}
w_{k}^{(\omega )}\doteq \left( 2\omega _{1}\left( 0\right) -\omega
_{1}\left( e_{k}\right) -\omega _{1}\left( -e_{k}\right) \right)
w_{k},\qquad k\in \{1,\ldots ,d\},
\end{equation*}%
with $\{e_{k}\}_{k=1}^{d}$ being the canonical basis of $\mathbb{R}^{d}$. By
(\ref{inequalty fluctuation 0}), (\ref{inequalty fluctuation 1}) and (\ref%
{inequalty fluctuation 2}), it suffices that 
\begin{equation*}
\mathbb{E}\left[ \left\vert \int_{-\infty }^{0}\left\langle w^{(\cdot )},%
\mathcal{E}\left( \alpha \right) \right\rangle _{\mathbb{R}^{d}}\alpha ^{2}%
\mathrm{d}\alpha \right\vert ^{2}\right] =\mathrm{Var}\left[ \int_{-\infty
}^{0}\left\langle w^{(\cdot )},\mathcal{E}\left( \alpha \right)
\right\rangle _{\mathbb{R}^{d}}\alpha ^{2}\mathrm{d}\alpha \right] >0
\end{equation*}%
in order to ensure non-vanishing quantum fluctuations of linear response
currents in the thermodynamic limit, i.e., $\partial _{s}^{2}\mathrm{J}^{(s%
\mathcal{E})}|_{s=0}\neq 0$.

Theorem \ref{new theorem} can be applied to the celebrated tight-binding
Anderson model, which corresponds to the special case $\vartheta =0$. This
is why we focus on this important example in this theorem. The remaining
case of larger parameters $\vartheta ,T\in \mathbb{R}_{0}^{+}$ can certainly
be studied, even if this is not done here.

\section{Technical Proofs\label{sec:proofs}}

\subsection{Quasi-Free Fermions in Subregions of the Lattice\label{Sectino
tech3 copy(2)}}

Let $\mathcal{P}_{\text{f}}(\mathbb{Z}^{d})\subseteq 2^{\mathbb{Z}^{d}}$ be
the set of all non-empty \emph{finite} subsets of $\mathbb{Z}^{d}$. Like in 
\cite[Section 2.1]{LDP}, we need the sets%
\begin{eqnarray*}
\mathfrak{Z} &\doteq &\left\{ \mathcal{Z}\subseteq 2^{\mathbb{Z}^{d}}\colon
\left( \forall Z_{1},Z_{2}\in \mathcal{Z}\right) \ Z_{1}\neq
Z_{2}\Rightarrow Z_{1}\cap Z_{2}=\emptyset \text{ }\right\} , \\
\mathfrak{Z}_{\text{f}} &\doteq &\mathfrak{Z}\cap \left\{ \mathcal{Z}%
\subseteq \mathcal{P}_{\text{f}}(\mathbb{Z}^{d})\colon \left\vert \mathcal{Z}%
\right\vert <\infty \right\} .
\end{eqnarray*}%
This kind of decomposition over collections of disjoint subsets of the
lattice is important to prove Theorem \ref{Quantum fluctuations and rate
function} (i).

Recall that $\mathfrak{h}\doteq \ell ^{2}(\mathbb{Z}^{d};\mathbb{C})$ and $%
\mathcal{B}(\mathfrak{h})$ is the Banach space of all bounded linear
operators acting on $\mathfrak{h}$. One can restrict the quasi-free dynamics
defined by (\ref{rescaledbis}) to collections $\mathcal{Z}\in \mathfrak{Z}$
of disjoint subsets of the lattice by using the orthogonal projections $%
P_{\Lambda }$, $\Lambda \subseteq \mathbb{Z}^{d}$, defined on the Hilbert
space $\mathfrak{h}$ by 
\begin{equation}
\lbrack P_{\Lambda }(\psi )](x)\doteq \left\{ 
\begin{array}{lll}
\psi (x) & , & \text{if }x\in \Lambda , \\ 
0 & , & \text{else,}%
\end{array}%
\right.  \label{orthogonal projection}
\end{equation}%
for any $\psi \in \mathfrak{h}$. Then, the one-particle Hamiltonian within $%
\mathcal{Z}\in \mathfrak{Z}$ is, by definition, equal to%
\begin{equation}
h_{\mathcal{Z}}^{(\omega )}\doteq \sum_{Z\in \mathcal{Z}}P_{Z}h^{(\omega
)}P_{Z}\in \mathcal{B}\left( \mathfrak{h}\right) ,  \label{h finite f}
\end{equation}%
where $h^{(\omega )}\in \mathcal{B}(\mathfrak{h})$ is the random
tight-binding model defined by (\ref{eq:Ham_lap_pot}) for any $\omega \in
\Omega $ and\ $\lambda ,\vartheta \in \mathbb{R}_{0}^{+}$. For any $\mathcal{%
Z}\in \mathfrak{Z}$, it leads to the unitary group $\{\mathrm{e}^{ith_{%
\mathcal{Z}}^{(\omega )}}\}_{t\in \mathbb{R}}$ acting on the Hilbert space $%
\mathfrak{h}$.

Similar to Equation (\ref{rescaledbis}), for any $\mathcal{Z}\in \mathfrak{Z}
$, we consequently define the strongly continuous group $\tau ^{(\omega ,%
\mathcal{Z})}\doteq \{\tau _{t}^{(\omega ,\mathcal{Z})}\}_{t\in {\mathbb{R}}%
} $ of Bogoliubov $\ast $-automorphisms of $\mathcal{U}$ by%
\begin{equation*}
\tau _{t}^{(\omega ,\mathcal{Z})}(a(\psi ))=a(\mathrm{e}^{ith_{\mathcal{Z}%
}^{(\omega )}}\psi )\ ,\text{\qquad }t\in \mathbb{R},\ \psi \in \mathfrak{h}.
\end{equation*}%
This corresponds to replace $h^{(\omega )}$ in (\ref{rescaledbis}) with $h_{%
\mathcal{Z}}^{(\omega )}$. Similarly, for any $\mathcal{Z}\in \mathfrak{Z}$,
we define the quasi-free state $\varrho _{\mathcal{Z}}^{(\omega )}$ by
replacing $h^{(\omega )}$ in Equation (\ref{2-point correlation function})
with the one-particle Hamiltonian $h_{\mathcal{Z}}^{(\omega )}$ within $%
\mathcal{Z}$.

If $\mathcal{Z}\in \mathfrak{Z}_{\text{f}}$ then both $\tau ^{(\omega ,%
\mathcal{Z})}$ and $\varrho _{\mathcal{Z}}^{(\omega )}$ can be written in
terms of bilinear elements\footnote{%
This refers to the well-known second-quantization of one-particle
Hamiltonians in the Fock space representation.}, defined as follows: The
bilinear element associated with an operator in $C\in \mathcal{B}(\mathfrak{h%
})$ whose range, $\mathrm{ran}(C)$, is finite dimensional is defined by 
\begin{equation}
\langle \mathrm{A},C\mathrm{A}\rangle \doteq \sum\limits_{i,j\in
I}\left\langle \psi _{i},C\psi _{j}\right\rangle _{\mathfrak{h}}a\left( \psi
_{i}\right) ^{\ast }a\left( \psi _{j}\right) ,  \label{bilinear_elements}
\end{equation}%
where $\{\psi _{i}\}_{i\in I}$ is any orthonormal basis\footnote{$\langle 
\mathrm{A},C\mathrm{A}\rangle $ does not depend on the particular choice of $%
\mathcal{H}$ and its orthonormal basis.} of a finite dimensional subspace 
\begin{equation*}
\mathcal{H}\supseteq \mathrm{ran}(C)\cup \mathrm{ran}(C^{\ast })
\end{equation*}%
of the Hilbert space $\mathfrak{h}$. See \cite[Definition 4.3]{LDP}. For any 
$\omega \in \Omega $ and $\lambda ,\vartheta \in \mathbb{R}_{0}^{+}$, the
range of $h_{\mathcal{Z}}^{(\omega )}\in \mathcal{B}(\mathfrak{h})$ is
finite-dimensional whenever $\mathcal{Z}\in \mathfrak{Z}_{\text{f}}$ and one
checks that, for any time $t\in \mathbb{R}$, inverse temperature $\beta \in 
\mathbb{R}^{+}$, finite collections $\mathcal{Z}\in \mathfrak{Z}_{\text{f}}$
and elements $B\in \mathcal{U}$,%
\begin{equation*}
\tau _{t}^{(\omega ,\mathcal{Z})}\left( B\right) =\mathrm{e}^{it\langle 
\mathrm{A},h_{\mathcal{Z}}^{(\omega )}\mathrm{A}\rangle }B\mathrm{e}%
^{-it\langle \mathrm{A},h_{\mathcal{Z}}^{(\omega )}\mathrm{A}\rangle }\qquad 
\text{and}\qquad \varrho _{\mathcal{Z}}^{(\omega )}(B)=\frac{\mathrm{tr}%
\left( B\mathrm{e}^{-\beta \langle \mathrm{A},h_{\mathcal{Z}}^{(\omega )}%
\mathrm{A}\rangle }\right) }{\mathrm{tr}\left( \mathrm{e}^{-\beta \langle 
\mathrm{A},h_{\mathcal{Z}}^{(\omega )}\mathrm{A}\rangle }\right) },
\end{equation*}%
where $\mathrm{tr}\in \mathcal{U}^{\ast }$ is the tracial state, i.e., the
gauge-invariant quasi-free state with two-point correlation functions given
by (\ref{2-point correlation function}) for $\beta =0$. See \cite[Equations
(27)-(28)]{LDP}. The dynamics corresponds in this case to the usual dynamics
written in the Heisenberg picture of quantum mechanics, while the above
quasi-free state is the Gibbs state at inverse temperature $\beta \in 
\mathbb{R}^{+}$, both associated with the Hamiltonian $\langle \mathrm{A},h_{%
\mathcal{Z}}^{(\omega )}\mathrm{A}\rangle \in \mathcal{U}$ for $\mathcal{Z}%
\in \mathfrak{Z}_{\text{f}}$.

In order to define the thermodynamic limit, we use the cubic boxes $\Lambda
_{\ell }\doteq \{\mathbb{Z}\cap \left[ -\ell ,\ell \right] \}^{d}$ for $\ell
\in \mathbb{R}_{0}^{+}$. Then, as $\ell \rightarrow \infty $, for any $t\in {%
\mathbb{R}}$, $\tau _{t}^{(\omega ,\{\Lambda _{\ell }\})}$ converges
strongly to $\tau _{t}^{(\omega )}\equiv \tau _{t}^{(\omega ,\{\mathbb{Z}%
^{d}\})}$, while $\varrho _{\{\Lambda _{\ell }\}}^{(\omega )}$ converges in
the weak$^{\ast }$ topology to $\varrho ^{(\omega )}\equiv \varrho _{\{%
\mathbb{Z}^{d}\}}^{(\omega )}$. For an explicit proof of these well-known
facts, see for instance \cite[Propositions 3.2.9 and 3.2.13]{Antsa-these}.

\subsection{Current Observables in Subregions of the Lattice\label{Sectino
tech3}}

Fix once and for all $\vec{w}\in {\mathbb{R}}^{d}$ with $\left\Vert \vec{w}%
\right\Vert _{\mathbb{R}^{d}}=1$. By \cite[Equation (29)]{LDP}, for any $%
\lambda ,\vartheta \in \mathbb{R}_{0}^{+}$, $\omega \in \Omega $, $\mathcal{E%
}\in C_{0}^{0}(\mathbb{R};\mathbb{R}^{d})$, $\mathcal{Z}\in \mathfrak{Z}_{%
\text{f}}$ and $\mathcal{Z}^{(\tau )}\in \mathfrak{Z}$, the linear response
current observable is, by definition, equal to%
\begin{eqnarray}
\mathfrak{K}_{\mathcal{Z},\mathcal{Z}^{(\tau )}}^{(\omega ,\mathcal{E})}
&\doteq &\underset{k,q=1}{\sum^{d}}w_{k}\sum_{Z\in \mathcal{Z}}\underset{%
x,y,x+e_{k},y+e_{q}\in Z}{\sum }\int_{-\infty }^{0}\left\{ \mathcal{E}\left(
\alpha \right) \right\} _{q}\mathrm{d}\alpha \int\nolimits_{0}^{-\alpha }%
\mathrm{d}s\ i[\tau _{-s}^{(\omega ,\mathcal{Z}^{(\tau )})}(I_{\left(
y+e_{q},y\right) }^{(\omega )}),I_{\left( x+e_{k},x\right) }^{(\omega )}] 
\notag \\
&&+2\underset{k=1}{\sum^{d}}w_{k}\sum_{Z\in \mathcal{Z}}\underset{%
x,x+e_{k}\in Z}{\sum }\left( \int_{-\infty }^{0}\left\{ \mathcal{E}\left(
\alpha \right) \right\} _{q}\mathrm{d}\alpha \right) \Re \mathrm{e}\left(
\langle \mathfrak{e}_{x+e_{k}},\Delta _{\omega ,\vartheta }\mathfrak{e}%
_{x}\rangle a(\mathfrak{e}_{x+e_{k}})^{\ast }a(\mathfrak{e}_{x})\right)
\label{eq:par-curr-sum-gen}
\end{eqnarray}%
with $\{e_{k}\}_{k=1}^{d}$ being the canonical basis of $\mathbb{R}^{d}$.
Recall that $\Re \mathrm{e}(A)\in \mathcal{U}$ is the real part of $A\in 
\mathcal{U}$, see (\ref{im and real part}). Note from Equations (\ref%
{defininion para coeff observable})-(\ref{current}) that%
\begin{equation}
\mathfrak{K}_{\{\Lambda \},\{\mathbb{Z}^{d}\}}^{(\omega ,\mathcal{E}%
)}=\left\vert \Lambda \right\vert \mathbb{I}_{\Lambda }^{(\omega ,\mathcal{E}%
)},\qquad \Lambda \in \mathcal{P}_{\text{f}}(\mathbb{Z}^{d}),  \label{aa}
\end{equation}%
are linear response current observables within finite subsets of the lattice.

The above current observables can obviously be rewritten as bilinear
elements (\ref{bilinear_elements}) associated with one-particle operators
acting on the Hilbert space $\mathfrak{h}$. In order to give an explicit
expression of these operators, we first define, for any $x\in \mathbb{Z}^{d}$%
, the shift operator $s_{x}\in \mathcal{B}(\mathfrak{h})$ by 
\begin{equation}
\left( s_{x}\psi \right) \left( y\right) \doteq \psi \left( x+y\right)
,\qquad y\in \mathbb{Z}^{d},\ \psi \in \mathfrak{h}.  \label{shift}
\end{equation}%
Note that $s_{x}^{\ast }=s_{-x}=s_{x}^{-1}$ for any $x\in \mathbb{Z}^{d}$.
Then, for every $\omega \in \Omega $ and $\vartheta \in \mathbb{R}_{0}^{+}$,
the single-hopping operators are%
\begin{equation}
S_{x,y}^{(\omega )}\doteq \langle \mathfrak{e}_{x},\Delta _{\omega
,\vartheta }\mathfrak{e}_{y}\rangle _{\mathfrak{h}}P_{\left\{ x\right\}
}s_{x-y}P_{\left\{ y\right\} },\qquad x,y\in \mathbb{Z}^{d},  \label{shift2}
\end{equation}%
where $P_{\left\{ u\right\} }$ is the orthogonal projection defined by (\ref%
{orthogonal projection}) for $\Lambda =\left\{ u\right\} $ and $u\in \mathbb{%
Z}^{d}$. Observe that 
\begin{equation*}
\left\langle \mathrm{A},S_{x,y}^{(\omega )}\mathrm{A}\right\rangle =\langle 
\mathfrak{e}_{x},\Delta _{\omega ,\vartheta }\mathfrak{e}_{y}\rangle _{%
\mathfrak{h}}a(\mathfrak{e}_{x})^{\ast }a(\mathfrak{e}_{y}),\qquad x,y\in 
\mathbb{Z}^{d}.
\end{equation*}%
Similarly, by the identity%
\begin{equation*}
\Im \mathrm{m}\left\{ \langle \mathrm{A},C\mathrm{A}\rangle \right\}
=\langle \mathrm{A},\Im \mathrm{m}\left\{ C\right\} \mathrm{A}\rangle
\end{equation*}%
for any $C\in \mathcal{B}(\mathfrak{h})$ whose range is finite dimensional,
the paramagnetic current observables defined by (\ref{current observable})
equals 
\begin{equation*}
I_{(x,y)}^{(\omega )}=-2\langle \mathrm{A},\Im \mathrm{m}\{S_{x,y}^{(\omega
)}\}\mathrm{A}\rangle ,\qquad x,y\in \mathbb{Z}^{d},
\end{equation*}%
for each $\omega \in \Omega $ and $\vartheta \in \mathbb{R}_{0}^{+}$. For
any $\lambda ,\vartheta \in \mathbb{R}_{0}^{+}$, $\omega \in \Omega $, $%
\mathcal{E}\in C_{0}^{0}(\mathbb{R};\mathbb{R}^{d})$, $\mathcal{Z}^{(\tau
)}\in \mathfrak{Z}$ and $\mathcal{Z}\in \mathfrak{Z}_{\text{f}}$, the
current observable (\ref{eq:par-curr-sum-gen}) can then be rewritten as 
\begin{equation}
\mathfrak{K}_{\mathcal{Z},\mathcal{Z}^{(\tau )}}^{(\omega ,\mathcal{E}%
)}=\left\langle \mathrm{A},K_{\mathcal{Z},\mathcal{Z}^{(\tau )}}^{(\omega ,%
\mathcal{E})}\mathrm{A}\right\rangle =\sum\limits_{x,y\in \mathbb{Z}%
^{d}}\left\langle \mathfrak{e}_{x},K_{\mathcal{Z},\mathcal{Z}^{(\tau
)}}^{(\omega ,\mathcal{E})}\mathfrak{e}_{y}\right\rangle _{\mathfrak{h}%
}a\left( \mathfrak{e}_{x}\right) ^{\ast }a\left( \mathfrak{e}_{y}\right) ,
\label{bilinear idiot}
\end{equation}%
where $K_{\mathcal{Z},\mathcal{Z}^{(\tau )}}^{(\omega ,\mathcal{E})}\in 
\mathcal{B}(\mathfrak{h})$ is the operator acting on the one-particle
Hilbert space $\mathfrak{h}$ defined by 
\begin{eqnarray}
K_{\mathcal{Z},\mathcal{Z}^{(\tau )}}^{(\omega ,\mathcal{E})} &\doteq &4%
\underset{k,q=1}{\sum^{d}}w_{k}\sum_{Z\in \mathcal{Z}}\underset{%
x,y,x+e_{k},y+e_{q}\in Z}{\sum }\int_{-\infty }^{0}\left\{ \mathcal{E}\left(
\alpha \right) \right\} _{q}\mathrm{d}\alpha  \notag \\
&&\qquad \qquad \qquad \int\nolimits_{0}^{-\alpha }\mathrm{d}s\ i\left[ 
\mathrm{e}^{-ish_{\mathcal{Z}^{(\tau )}}^{(\omega )}}\Im \mathrm{m}%
\{S_{y+e_{q},y}^{(\omega )}\}\mathrm{e}^{ish_{\mathcal{Z}^{(\tau
)}}^{(\omega )}},\Im \mathrm{m}\{S_{x+e_{k},x}^{(\omega )}\}\right]  \notag
\\
&&+2\underset{k=1}{\sum^{d}}w_{k}\sum_{Z\in \mathcal{Z}}\underset{%
x,x+e_{k}\in Z}{\sum }\left( \int_{-\infty }^{0}\left\{ \mathcal{E}\left(
\alpha \right) \right\} _{q}\mathrm{d}\alpha \right) \Re \mathrm{e}%
\{S_{x+e_{k},x}^{(\omega )}\}.  \label{definition K}
\end{eqnarray}%
Note that the range of this bounded and self-adjoint operator is
finite-dimensional whenever $\mathcal{Z}\in \mathfrak{Z}_{\text{f}}$.

\subsection{Differentiability Class of Generating Functions\label{Sectino
tech3 copy(1)}}

The aim of this section is to prove Theorem \ref{Quantum fluctuations and
rate function} (i), in particular that the generating function $s\mapsto 
\mathrm{J}^{(s\mathcal{E})}$ defined by (\ref{generating_function}) belongs
to $C^{2}\left( \mathbb{R};\mathbb{R}\right) $. By \cite[Theorem 3.1]{LDP},
we already know that it is a well-defined, continuously differentiable,
convex function. So, one has to prove here that\ the second derivative of
the generating function exists and is continuous. To arrive at this
assertion, we follow the lines of arguments of \cite[Section 4]{LDP} showing 
\cite[Theorem 3.1]{LDP} via the control of the thermodynamic limit of
finite-volume generating functions that are random.

Fix once and for all $\beta \in \mathbb{R}^{+}$, $\lambda ,\vartheta \in 
\mathbb{R}_{0}^{+}$ and $\vec{w}\in {\mathbb{R}}^{d}$ with $\left\Vert \vec{w%
}\right\Vert _{\mathbb{R}^{d}}=1$. For any $\mathcal{E}\in C_{0}^{0}(\mathbb{%
R};\mathbb{R}^{d})$, $\omega \in \Omega $ and three finite collections $%
\mathcal{Z},\mathcal{Z}^{(\varrho )},\mathcal{Z}^{(\tau )}\in \mathfrak{Z}_{%
\text{f}}$, we define the finite-volume generating function 
\begin{equation}
\mathrm{J}_{\mathcal{Z},\mathcal{Z}^{(\varrho )},\mathcal{Z}^{(\tau
)}}^{(\omega ,\mathcal{E})}\doteq g_{\mathcal{Z},\mathcal{Z}^{(\varrho )},%
\mathcal{Z}^{(\tau )}}^{(\omega ,\mathcal{E})}-g_{\mathcal{Z},\mathcal{Z}%
^{(\varrho )},\mathcal{Z}^{(\tau )}}^{(\omega ,0)},
\label{generating functions0}
\end{equation}%
where%
\begin{equation}
g_{\mathcal{Z},\mathcal{Z}^{(\varrho )},\mathcal{Z}^{(\tau )}}^{(\omega ,%
\mathcal{E})}\doteq \frac{1}{|\cup \mathcal{Z}|}\ln \mathrm{tr}\left( \exp
(-\beta \langle \mathrm{A},h_{\mathcal{Z}^{(\varrho )}}^{(\omega )}\mathrm{A}%
\rangle )\exp (\mathfrak{K}_{\mathcal{Z},\mathcal{Z}^{(\tau )}}^{(\omega ,%
\mathcal{E})})\right) .  \label{generating functions}
\end{equation}%
Recall that the tracial state $\mathrm{tr}\in \mathcal{U}^{\ast }$ is the
gauge-invariant quasi-free state with two-point correlation function given
by (\ref{2-point correlation function}) for $\beta =0$, while $h_{\mathcal{Z}%
^{(\varrho )}}^{(\omega )}$ is the one-particle Hamiltonian defined by (\ref%
{h finite f}). See also (\ref{bilinear_elements}) and (\ref%
{eq:par-curr-sum-gen}). Compare (\ref{generating functions0})-(\ref%
{generating functions}) with the equalities%
\begin{eqnarray}
\mathrm{J}^{(\mathcal{E})} &\doteq &\lim_{L\rightarrow \infty }\frac{1}{%
\left\vert \Lambda _{L}\right\vert }\mathbb{E}\left[ \ln \varrho ^{(\cdot
)}\left( \mathrm{e}^{\left\vert \Lambda _{L}\right\vert \mathbb{I}_{\Lambda
_{L}}^{(\cdot ,\mathcal{E})}}\right) \right]  \notag \\
&=&\lim_{L\rightarrow \infty }\frac{1}{\left\vert \Lambda _{L}\right\vert }%
\ln \varrho ^{(\omega )}\left( \mathrm{e}^{\left\vert \Lambda
_{L}\right\vert \mathbb{I}_{\Lambda _{L}}^{(\omega ,\mathcal{E})}}\right)
=\lim_{L\rightarrow \infty }\lim_{L_{\varrho }\rightarrow \infty
}\lim_{L_{\tau }\rightarrow \infty }\mathrm{J}_{\mathcal{\{}\Lambda _{L}%
\mathcal{\}},\mathcal{\{}\Lambda _{L_{\varrho }}\mathcal{\}},\mathcal{\{}%
\Lambda _{L_{\tau }}\mathcal{\}}}^{(\omega ,\mathcal{E})},
\label{limit cool}
\end{eqnarray}%
where the random variable $\omega $ is in a measurable subset of full measure%
\footnote{%
The measurable subset $\tilde{\Omega}\subseteq \Omega $ of full measure of 
\cite[Theorem 3.1]{LDP} does not depend on $\beta \in \mathbb{R}^{+}$, $%
\vartheta ,\lambda \in \mathbb{R}_{0}^{+}$, $\mathcal{E}\in C_{0}^{0}(%
\mathbb{R};\mathbb{R}^{d})$ and $\vec{w}\in {\mathbb{R}}^{d}$ with $%
\left\Vert \vec{w}\right\Vert _{\mathbb{R}^{d}}=1$.}, by \cite[Theorem 3.1
and Equation (45)]{LDP}. Recall that $\Lambda _{\ell }\doteq \{\mathbb{Z}%
\cap \left[ -\ell ,\ell \right] \}^{d}$ for $\ell \in \mathbb{R}_{0}^{+}$.
(See again (\ref{generating_function}) for the definition of the generating
function.) In fact, by \cite[Proposition 4.10]{LDP}, the above local
generating functions can be approximately decomposed into boxes of fixed
volume and we use the Ackoglu-Krengel (superadditive) ergodic theorem \cite[%
Theorem 4.17]{LDP} to deduce, via \cite[Proposition 4.8]{LDP}, the existence
of the generating functions as the thermodynamic limit of finite-volume
generating functions, as given in (\ref{limit cool}).

In order to prove that the generating function is continuously
differentiable, one uses in \cite[Corollary 4.20]{LDP} the (Arzel\`{a}-)
Ascoli theorem \cite[Theorem A5]{Rudin}. This approach requires uniform
bounds on the first and second derivatives of the finite-volume generating
functions 
\begin{equation}
s\mapsto \mathrm{J}_{\mathcal{Z},\mathcal{Z}^{(\varrho )},\mathcal{Z}^{(\tau
)}}^{(\omega ,s\mathcal{E})},\qquad \mathcal{E}\in C_{0}^{0}(\mathbb{R};%
\mathbb{R}^{d}),\ \omega \in \Omega ,\ \mathcal{Z},\mathcal{Z}^{(\varrho )},%
\mathcal{Z}^{(\tau )}\in \mathfrak{Z}_{\text{f}}.  \label{kl}
\end{equation}%
This is done in \cite[Proposition 4.9]{LDP}, which establishes the
following: Fixing $\mathcal{E}\in C_{0}^{0}\left( \mathbb{R};\mathbb{R}%
^{d}\right) $ and $\beta _{1},s_{1},\vartheta _{1},\lambda _{1}\in \mathbb{R}%
^{+}$, one has%
\begin{equation}
\sup_{\substack{ \beta \in \left( 0,\beta _{1}\right] ,\ \vartheta \in \left[
0,\vartheta _{1}\right] ,\ \lambda \in \left[ 0,\lambda _{1}\right]  \\ %
\omega \in \Omega ,\ s\in \left[ -s_{1},s_{1}\right] ,\ \mathcal{Z},\mathcal{%
Z}^{(\varrho )},\mathcal{Z}^{(\tau )}\in \mathfrak{Z}_{\text{f}}}}\left\{
\left\vert \partial _{s}\mathrm{J}_{\mathcal{Z},\mathcal{Z}^{(\varrho )},%
\mathcal{Z}^{(\tau )}}^{(\omega ,s\mathcal{E})}\right\vert +\left\vert
\partial _{s}^{2}\mathrm{J}_{\mathcal{Z},\mathcal{Z}^{(\varrho )},\mathcal{Z}%
^{(\tau )}}^{(\omega ,s\mathcal{E})}\right\vert \right\} <\infty .
\label{proposition 4.9 new}
\end{equation}%
In order to get in the same way the existence and continuity of the second
derivative of the generating function, we need now to control the \emph{third%
}-order derivative of the same finite-volume generating functions (\ref{kl}).

Equation (\ref{proposition 4.9 new}) is proven by using the CAR (\ref{CAR})
and the Combes-Thomas estimate \cite[Appendix A]{LDP}, in particular the
bound%
\begin{equation}
\sup_{\lambda \in \mathbb{R}_{0}^{+}}\sup_{\mathcal{Z}\in \mathfrak{Z}%
}\sup_{\omega \in \Omega }\left\vert \left\langle \mathfrak{e}_{x},\mathrm{e}%
^{ith_{\mathcal{Z}}^{(\omega )}}\mathfrak{e}_{y}\right\rangle _{\mathfrak{h}%
}\right\vert \leq 36\mathrm{e}^{\left\vert t\eta \right\vert -2\mu _{\eta
}|x-y|},\qquad x,y\in \mathbb{Z}^{d},\ \vartheta \in \mathbb{R}_{0}^{+},\
t\in \mathbb{R},  \label{Combes-ThomasCombes-Thomas}
\end{equation}%
(see \cite[Equation (7)]{LDP}), where 
\begin{equation}
\mu _{\eta }\doteq \mu \min \left\{ \frac{1}{2},\frac{\eta }{8d\left(
1+\vartheta \right) \mathrm{e}^{\mu }}\right\} ,
\label{Combes-ThomasCombes-Thomasbis}
\end{equation}%
the parameters $\eta ,\mu \in \mathbb{R}^{+}$ being two arbitrarily fixed
(strictly positive) constants. For any $\mathcal{E}\in C_{0}^{0}(\mathbb{R};%
\mathbb{R}^{d})$ and $\beta _{1},s_{1},\vartheta _{1},\lambda _{1}\in 
\mathbb{R}^{+}$, the Combes-Thomas estimate leads also to the uniform
estimates 
\begin{equation}
\sup_{\substack{ \beta \in \left( 0,\beta _{1}\right] ,\ \vartheta \in \left[
0,\vartheta _{1}\right] ,\ \lambda \in \left[ 0,\lambda _{1}\right]  \\ %
\omega \in \Omega ,\ s\in \left[ -s_{1},s_{1}\right] ,\ \mathcal{Z},\mathcal{%
Z}^{(\varrho )},\mathcal{Z}^{(\tau )}\in \mathfrak{Z}_{\text{f}}}}\sup_{x\in 
\mathbb{Z}^{d}}\sum_{y\in \mathbb{Z}^{d}}\left\vert \left\langle \mathfrak{e}%
_{y},\frac{1}{1+\mathrm{e}^{-\frac{s}{2}K_{\mathcal{Z},\mathcal{Z}^{(\tau
)}}^{(\omega ,\mathcal{E})}}\mathrm{e}^{\beta h_{\mathcal{Z}^{(\varrho
)}}^{(\omega )}}\mathrm{e}^{-\frac{s}{2}K_{\mathcal{Z},\mathcal{Z}^{(\tau
)}}^{(\omega ,\mathcal{E})}}}\mathfrak{e}_{x}\right\rangle _{\mathfrak{h}%
}\right\vert <\infty  \label{fghfh}
\end{equation}%
(see the end of the proof of \cite[Proposition 4.9]{LDP}) as well as 
\begin{equation}
\sup_{\vartheta \in \left[ 0,\vartheta _{1}\right] }\sup_{\lambda \in 
\mathbb{R}_{0}^{+}}\sup_{\mathcal{Z},\mathcal{Z}^{(\tau )}\in \mathfrak{Z}_{%
\text{f}}}\sup_{\omega \in \Omega }\frac{1}{|\cup \mathcal{Z}|}%
\sum\limits_{x,y\in \mathbb{Z}^{d}}\left\vert \left\langle \mathfrak{e}%
_{x},K_{\mathcal{Z},\mathcal{Z}^{(\tau )}}^{(\omega ,\mathcal{E})}\mathfrak{e%
}_{y}\right\rangle _{\mathfrak{h}}\right\vert <\infty
\end{equation}%
and 
\begin{equation}
\sup_{\vartheta \in \left[ 0,\vartheta _{1}\right] }\sup_{\lambda \in 
\mathbb{R}_{0}^{+}}\sup_{\mathcal{Z},\mathcal{Z}^{(\tau )}\in \mathfrak{Z}_{%
\text{f}}}\sup_{\omega \in \Omega }\left\vert \left\langle \mathfrak{e}%
_{x},K_{\mathcal{Z},\mathcal{Z}^{(\tau )}}^{(\omega ,\mathcal{E})}\mathfrak{e%
}_{y}\right\rangle _{\mathfrak{h}}\right\vert \leq C_{x,y}^{(\mathcal{E}%
,\vartheta _{1})}<\infty  \label{fghfh0}
\end{equation}%
for $x,y\in \mathbb{Z}^{d}$, where $C_{x,y}^{(\mathcal{E},\vartheta
_{1})}\in \mathbb{R}^{+}$ are constants satisfying 
\begin{equation}
\sup_{x,y\in \mathbb{Z}^{d}}C_{x,y}^{(\mathcal{E},\vartheta _{1})}<\infty
\qquad \text{and}\qquad \sup_{x\in \mathbb{Z}^{d}}\sum_{y\in \mathbb{Z}%
^{d}}C_{x,y}^{(\mathcal{E},\vartheta _{1})}<\infty .
\label{corollaire utile}
\end{equation}%
Recall that $K_{\mathcal{Z},\mathcal{Z}^{(\tau )}}^{(\omega ,\mathcal{E}%
)}\in \mathcal{B}(\mathfrak{h})$ is the operator defining linear response
current observables, by (\ref{bilinear idiot})-(\ref{definition K}).

In order to give a uniform estimate on the third-order derivative of the
finite-volume generating functions (\ref{kl}), similar to the proof of
Equation (\ref{proposition 4.9 new}), we use again the Combes-Thomas
estimate, which yields (\ref{fghfh})-(\ref{corollaire utile}). This proof
bears however on more complex computations than the one of Equation (\ref%
{proposition 4.9 new}), which only controls the first and second derivatives
of the same function.

\begin{proposition}[Uniform boundedness of third derivatives]
\label{3eme_derivee}\mbox{}\newline
Fix an electric field $\mathcal{E}\in C_{0}^{0}\left( \mathbb{R};\mathbb{R}%
^{d}\right) $, $\vec{w}\in {\mathbb{R}}^{d}$ with $\left\Vert \vec{w}%
\right\Vert _{\mathbb{R}^{d}}=1$ and the parameters $\beta
_{1},s_{1},\vartheta _{1},\lambda _{1}\in \mathbb{R}^{+}$. Then, 
\begin{equation*}
\sup_{\substack{ \beta \in \left( 0,\beta _{1}\right] ,\ \vartheta \in \left[
0,\vartheta _{1}\right] ,\ \lambda \in \left[ 0,\lambda _{1}\right]  \\ %
\omega \in \Omega ,\ s\in \left[ -s_{1},s_{1}\right] ,\ \mathcal{Z},\mathcal{%
Z}^{(\varrho )},\mathcal{Z}^{(\tau )}\in \mathfrak{Z}_{\mathrm{f}}}}%
\left\vert \partial _{s}^{3}\mathrm{J}_{\mathcal{Z},\mathcal{Z}^{(\varrho )},%
\mathcal{Z}^{(\tau )}}^{(\omega ,s\mathcal{E})}\right\vert <\infty .
\end{equation*}
\end{proposition}

\begin{proof}
For any $\beta \in \mathbb{R}^{+}$, $\vartheta ,\lambda \in \mathbb{R}%
_{0}^{+}$, $\mathcal{E}\in C_{0}^{0}(\mathbb{R};\mathbb{R}^{d})$, $\vec{w}%
\in {\mathbb{R}}^{d}$ with $\left\Vert \vec{w}\right\Vert _{\mathbb{R}%
^{d}}=1 $ and $\mathcal{Z},\mathcal{Z}^{(\varrho )},\mathcal{Z}^{(\tau )}\in 
\mathfrak{Z}_{\mathrm{f}}$, a straightforward computation yields that 
\begin{eqnarray}
&&\partial _{s}^{3}\mathrm{J}_{\mathcal{Z},\mathcal{Z}^{(\varrho )},\mathcal{%
Z}^{(\tau )}}^{(\omega ,s\mathcal{E})}  \label{dfdfdfdfd} \\
&=&\frac{1}{|\cup \mathcal{Z}|}\varpi _{s}^{T}\left( \mathfrak{K}_{\mathcal{Z%
},\mathcal{Z}^{(\tau )}}^{(\omega ,\mathcal{E})};\mathfrak{K}_{\mathcal{Z},%
\mathcal{Z}^{(\tau )}}^{(\omega ,\mathcal{E})};\mathfrak{K}_{\mathcal{Z},%
\mathcal{Z}^{(\tau )}}^{(\omega ,\mathcal{E})}\right)  \notag \\
&=&\frac{1}{|\cup \mathcal{Z}|}\left( \varpi _{s}\left( \left( \mathfrak{K}_{%
\mathcal{Z},\mathcal{Z}^{(\tau )}}^{(\omega ,\mathcal{E})}\right)
^{3}\right) -3\varpi _{s}\left( \left( \mathfrak{K}_{\mathcal{Z},\mathcal{Z}%
^{(\tau )}}^{(\omega ,\mathcal{E})}\right) ^{2}\right) \varpi _{s}\left( 
\mathfrak{K}_{\mathcal{Z},\mathcal{Z}^{(\tau )}}^{(\omega ,\mathcal{E}%
)}\right) +2\varpi _{s}\left( \mathfrak{K}_{\mathcal{Z},\mathcal{Z}^{(\tau
)}}^{(\omega ,\mathcal{E})}\right) ^{3}\right) ,  \notag
\end{eqnarray}%
where $\varpi _{s}$ is the (unique) gauge-invariant quasi-free state
satisfying 
\begin{equation}
\varpi _{s}(a^{\ast }\left( \varphi \right) a\left( \psi \right)
)=\left\langle \psi ,\frac{1}{1+\mathrm{e}^{-\frac{s}{2}K_{\mathcal{Z},%
\mathcal{Z}^{(\tau )}}^{(\omega ,\mathcal{E})}}\mathrm{e}^{\beta h_{\mathcal{%
Z}^{(\varrho )}}^{(\omega )}}\mathrm{e}^{-\frac{s}{2}K_{\mathcal{Z},\mathcal{%
Z}^{(\tau )}}^{(\omega ,\mathcal{E})}}}\varphi \right\rangle _{\mathfrak{h}%
},\qquad \varphi ,\psi \in \mathfrak{h}.  \label{cool2}
\end{equation}%
In the first equality of (\ref{dfdfdfdfd}), $\varpi _{s}^{T}(\cdot ;\cdot
;\cdot )$ denotes the so-called \textquotedblleft
truncated\textquotedblright\ or \textquotedblleft
connected\textquotedblright\ correlation function of third order, associated
with the state $\varpi _{s}$. Recall that, for all $A_{1},A_{2},A_{3}\in 
\mathcal{U}$, this function is defined by%
\begin{eqnarray*}
\varpi _{s}^{T}(A_{1};A_{2};A_{3}) &\doteq &\varpi
_{s}(A_{1}A_{2}A_{3})-\varpi _{s}(A_{1})\varpi _{s}(A_{2}A_{3})-\varpi
_{s}(A_{2})\varpi _{s}(A_{1}A_{3}) \\
&&-\varpi _{s}(A_{3})\varpi _{s}(A_{1}A_{2})+2\varpi _{s}(A_{1})\varpi
_{s}(A_{2})\varpi _{s}(A_{3}).
\end{eqnarray*}%
(This is similar to \cite[Proof of Proposition 4.9, until Equation (48)]{LDP}%
.) Recall that $\left\{ \mathfrak{e}_{x}\right\} _{x\in \mathbb{Z}^{d}}$ is
the canonical orthonormal basis of $\mathfrak{h}$, which is defined by $%
\mathfrak{e}_{x}(y)\doteq \delta _{x,y}$ for all $x,y\in \mathbb{Z}^{d}$. By
linearity and continuity in each argument of $\varpi _{s}^{T}(\cdot ;\cdot
;\cdot )$, one has%
\begin{eqnarray*}
\partial _{s}^{3}\mathrm{J}_{\mathcal{Z},\mathcal{Z}^{(\varrho )},\mathcal{Z}%
^{(\tau )}}^{(\omega ,s\mathcal{E})} &=&\frac{1}{|\cup \mathcal{Z}|}%
\sum\limits_{x_{1},y_{1},x_{2},y_{2},x_{3},y_{3}\in \mathbb{Z}%
^{d}}\left\langle \mathfrak{e}_{x_{1}},K_{\mathcal{Z},\mathcal{Z}^{(\tau
)}}^{(\omega ,\mathcal{E})}\mathfrak{e}_{y_{1}}\right\rangle _{\mathfrak{h}%
}\left\langle \mathfrak{e}_{x_{2}},K_{\mathcal{Z},\mathcal{Z}^{(\tau
)}}^{(\omega ,\mathcal{E})}\mathfrak{e}_{y_{2}}\right\rangle _{\mathfrak{h}%
}\left\langle \mathfrak{e}_{x_{3}},K_{\mathcal{Z},\mathcal{Z}^{(\tau
)}}^{(\omega ,\mathcal{E})}\mathfrak{e}_{y_{3}}\right\rangle _{\mathfrak{h}}
\\
&&\times \varpi _{s}^{T}(a^{\ast }\left( \mathfrak{e}_{x_{1}}\right) a\left( 
\mathfrak{e}_{y_{1}}\right) ;a^{\ast }\left( \mathfrak{e}_{x_{2}}\right)
a\left( \mathfrak{e}_{y_{2}}\right) ;a^{\ast }\left( \mathfrak{e}%
_{x_{3}}\right) a\left( \mathfrak{e}_{y_{3}}\right) ).
\end{eqnarray*}%
Note that, by Equation (\ref{ass O0-00bis}) and the fact that $\varpi _{s}$
is a gauge-invariant quasi-free state, 
\begin{eqnarray*}
&&\varpi _{s}(a^{\ast }\left( \mathfrak{e}_{x_{1}}\right) a\left( \mathfrak{e%
}_{y_{1}}\right) a^{\ast }\left( \mathfrak{e}_{x_{2}}\right) a\left( 
\mathfrak{e}_{y_{2}}\right) a^{\ast }\left( \mathfrak{e}_{x_{3}}\right)
a\left( \mathfrak{e}_{y_{3}}\right) ) \\
&=&\mathrm{det}\left( 
\begin{array}{ccc}
\varpi _{s}\left( a^{\ast }(\mathfrak{e}_{x_{1}})a(\mathfrak{e}%
_{y_{1}})\right) & \varpi _{s}\left( a^{\ast }(\mathfrak{e}_{x_{1}})a(%
\mathfrak{e}_{y_{2}})\right) & \varpi _{s}\left( a^{\ast }(\mathfrak{e}%
_{x_{1}})a(\mathfrak{e}_{y_{3}})\right) \\ 
-\varpi _{s}\left( a(\mathfrak{e}_{y_{1}})a^{\ast }(\mathfrak{e}%
_{x_{2}})\right) & \varpi _{s}\left( a^{\ast }(\mathfrak{e}_{x_{2}})a(%
\mathfrak{e}_{y_{2}})\right) & \varpi _{s}\left( a^{\ast }(\mathfrak{e}%
_{x_{2}})a(\mathfrak{e}_{y_{3}})\right) \\ 
-\varpi _{s}\left( a(\mathfrak{e}_{y_{1}})a^{\ast }(\mathfrak{e}%
_{x_{3}})\right) & -\varpi _{s}\left( a(\mathfrak{e}_{y_{2}})a^{\ast }(%
\mathfrak{e}_{x_{3}})\right) & \varpi _{s}\left( a^{\ast }(\mathfrak{e}%
_{x_{3}})a(\mathfrak{e}_{y_{3}})\right)%
\end{array}%
\right) \\
&=&\sum\limits_{g\in \mathcal{G}_{3}}\xi _{s}^{g}\left(
x_{1},y_{1},x_{2},y_{2},x_{3},y_{3}\right)
\end{eqnarray*}%
(use, for instance, \cite[Lemma 3.1]{universal} to get the above
determinant), where 
\begin{eqnarray*}
\mathcal{G}_{3} &\doteq
&\{\{(1,1),(2,2),(3,3)\},\{(1,1),(2,3),(3,2)\},\{(1,2),(2,1),(3,3)\}\} \\
&&\cup \{\{(1,2),(2,3),(3,1)\},\{(1,3),(2,1),(3,2)\},\{(1,3),(2,2),(3,1)\}\}
\end{eqnarray*}%
is a set of oriented graphs with vertex set $\{1,2,3\}$ and%
\begin{eqnarray*}
\xi _{s}^{\{(1,1),(2,2),(3,3)\}}\left(
x_{1},y_{1},x_{2},y_{2},x_{3},y_{3}\right) &\doteq &\varpi _{s}\left(
a^{\ast }(\mathfrak{e}_{x_{1}})a(\mathfrak{e}_{y_{1}})\right) \varpi
_{s}\left( a^{\ast }(\mathfrak{e}_{x_{2}})a(\mathfrak{e}_{y_{2}})\right)
\varpi _{s}\left( a^{\ast }(\mathfrak{e}_{x_{3}})a(\mathfrak{e}%
_{y_{3}})\right) , \\
\xi _{s}^{\{(1,1),(2,3),(3,2)\}}\left(
x_{1},y_{1},x_{2},y_{2},x_{3},y_{3}\right) &\doteq &\varpi _{s}\left(
a^{\ast }(\mathfrak{e}_{x_{1}})a(\mathfrak{e}_{y_{1}})\right) \varpi
_{s}\left( a^{\ast }(\mathfrak{e}_{x_{2}})a(\mathfrak{e}_{y_{3}})\right)
\varpi _{s}\left( a(\mathfrak{e}_{y_{2}})a^{\ast }(\mathfrak{e}%
_{x_{3}})\right) , \\
\xi _{s}^{\{(1,2),(2,1),(3,3)\}}\left(
x_{1},y_{1},x_{2},y_{2},x_{3},y_{3}\right) &\doteq &\varpi _{s}\left(
a^{\ast }(\mathfrak{e}_{x_{1}})a(\mathfrak{e}_{y_{2}})\right) \varpi
_{s}\left( a(\mathfrak{e}_{y_{1}})a^{\ast }(\mathfrak{e}_{x_{2}})\right)
\varpi _{s}\left( a^{\ast }(\mathfrak{e}_{x_{3}})a(\mathfrak{e}%
_{y_{3}})\right) , \\
\xi _{s}^{\{(1,2),(2,3),(3,1)\}}\left(
x_{1},y_{1},x_{2},y_{2},x_{3},y_{3}\right) &\doteq &-\varpi _{s}\left(
a^{\ast }(\mathfrak{e}_{x_{1}})a(\mathfrak{e}_{y_{2}})\right) \varpi
_{s}\left( a^{\ast }(\mathfrak{e}_{x_{2}})a(\mathfrak{e}_{y_{3}})\right)
\varpi _{s}\left( a(\mathfrak{e}_{y_{1}})a^{\ast }(\mathfrak{e}%
_{x_{3}})\right) , \\
\xi _{s}^{\{(1,3),(2,1),(3,2)\}}\left(
x_{1},y_{1},x_{2},y_{2},x_{3},y_{3}\right) &\doteq &\varpi _{s}\left(
a^{\ast }(\mathfrak{e}_{x_{1}})a(\mathfrak{e}_{y_{3}})\right) \varpi
_{s}\left( a(\mathfrak{e}_{y_{1}})a^{\ast }(\mathfrak{e}_{x_{2}})\right)
\varpi _{s}\left( a(\mathfrak{e}_{y_{2}})a^{\ast }(\mathfrak{e}%
_{x_{3}})\right) , \\
\xi _{s}^{\{(1,3),(2,2),(3,1)\}}\left(
x_{1},y_{1},x_{2},y_{2},x_{3},y_{3}\right) &\doteq &\varpi _{s}\left(
a^{\ast }(\mathfrak{e}_{x_{1}})a(\mathfrak{e}_{y_{3}})\right) \varpi
_{s}\left( a^{\ast }(\mathfrak{e}_{x_{2}})a(\mathfrak{e}_{y_{2}})\right)
\varpi _{s}\left( a(\mathfrak{e}_{y_{1}})a^{\ast }(\mathfrak{e}%
_{x_{3}})\right) .
\end{eqnarray*}%
By elementary computations, one sees that taking connected correlations
corresponds here, as is usual, to only keep the terms associated with
connected graphs. That is,%
\begin{multline*}
\varpi _{s}^{T}(a^{\ast }\left( \mathfrak{e}_{x_{1}}\right) a\left( 
\mathfrak{e}_{y_{1}}\right) ;a^{\ast }\left( \mathfrak{e}_{x_{2}}\right)
a\left( \mathfrak{e}_{y_{2}}\right) ;a^{\ast }\left( \mathfrak{e}%
_{x_{3}}\right) a\left( \mathfrak{e}_{y_{3}}\right) ) \\
=\varpi _{s}\left( a^{\ast }(\mathfrak{e}_{x_{1}})a(\mathfrak{e}%
_{y_{3}})\right) \varpi _{s}\left( a(\mathfrak{e}_{y_{1}})a^{\ast }(%
\mathfrak{e}_{x_{2}})\right) \varpi _{s}\left( a(\mathfrak{e}%
_{y_{2}})a^{\ast }(\mathfrak{e}_{x_{3}})\right) \\
-\varpi _{s}\left( a^{\ast }(\mathfrak{e}_{x_{1}})a(\mathfrak{e}%
_{y_{2}})\right) \varpi _{s}\left( a^{\ast }(\mathfrak{e}_{x_{2}})a(%
\mathfrak{e}_{y_{3}})\right) \varpi _{s}\left( a(\mathfrak{e}%
_{y_{1}})a^{\ast }(\mathfrak{e}_{x_{3}})\right) .
\end{multline*}%
Hence, 
\begin{equation}
\partial _{s}^{3}\mathrm{J}_{\mathcal{Z},\mathcal{Z}^{(\varrho )},\mathcal{Z}%
^{(\tau )}}^{(\omega ,s\mathcal{E})}=\mathbf{K}_{1}-\mathbf{K}_{2},
\label{K0}
\end{equation}%
where 
\begin{align}
\mathbf{K}_{1}& \doteq \frac{1}{|\cup \mathcal{Z}|}\sum%
\limits_{x_{1},y_{1},x_{2},y_{2},x_{3},y_{3}\in \mathbb{Z}^{d}}\left\langle 
\mathfrak{e}_{x_{1}},K_{\mathcal{Z},\mathcal{Z}^{(\tau )}}^{(\omega ,%
\mathcal{E})}\mathfrak{e}_{y_{1}}\right\rangle _{\mathfrak{h}}\left\langle 
\mathfrak{e}_{x_{2}},K_{\mathcal{Z},\mathcal{Z}^{(\tau )}}^{(\omega ,%
\mathcal{E})}\mathfrak{e}_{y_{2}}\right\rangle _{\mathfrak{h}}\left\langle 
\mathfrak{e}_{x_{3}},K_{\mathcal{Z},\mathcal{Z}^{(\tau )}}^{(\omega ,%
\mathcal{E})}\mathfrak{e}_{y_{3}}\right\rangle _{\mathfrak{h}}  \label{K11}
\\
& \qquad \qquad \qquad \qquad \varpi _{s}\left( a^{\ast }\left( \mathfrak{e}%
_{x_{1}}\right) a\left( \mathfrak{e}_{y_{3}}\right) \right) \varpi
_{s}\left( a\left( \mathfrak{e}_{y_{1}}\right) a^{\ast }\left( \mathfrak{e}%
_{x_{2}}\right) \right) \varpi _{s}\left( a\left( \mathfrak{e}%
_{y_{2}}\right) a^{\ast }\left( \mathfrak{e}_{x_{3}}\right) \right)  \notag
\end{align}%
and 
\begin{align}
\mathbf{K}_{2}& \doteq \frac{1}{|\cup \mathcal{Z}|}\sum%
\limits_{x_{1},y_{1},x_{2},y_{2},x_{3},y_{3}\in \mathbb{Z}^{d}}\left\langle 
\mathfrak{e}_{x_{1}},K_{\mathcal{Z},\mathcal{Z}^{(\tau )}}^{(\omega ,%
\mathcal{E})}\mathfrak{e}_{y_{1}}\right\rangle _{\mathfrak{h}}\left\langle 
\mathfrak{e}_{x_{2}},K_{\mathcal{Z},\mathcal{Z}^{(\tau )}}^{(\omega ,%
\mathcal{E})}\mathfrak{e}_{y_{2}}\right\rangle _{\mathfrak{h}}\left\langle 
\mathfrak{e}_{x_{3}},K_{\mathcal{Z},\mathcal{Z}^{(\tau )}}^{(\omega ,%
\mathcal{E})}\mathfrak{e}_{y_{3}}\right\rangle _{\mathfrak{h}}  \label{K2} \\
& \qquad \qquad \qquad \qquad \varpi _{s}\left( a^{\ast }\left( \mathfrak{e}%
_{x_{1}}\right) a\left( \mathfrak{e}_{y_{2}}\right) \right) \varpi
_{s}\left( a\left( \mathfrak{e}_{y_{1}}\right) a^{\ast }\left( \mathfrak{e}%
_{x_{3}}\right) \right) \varpi _{s}\left( a^{\ast }\left( \mathfrak{e}%
_{x_{2}}\right) a\left( \mathfrak{e}_{y_{3}}\right) \right) .  \notag
\end{align}%
Applying the triangle inequality, we now obtain that 
\begin{multline*}
|\mathbf{K}_{1}|\leq \frac{1}{|\cup \mathcal{Z}|}\sum%
\limits_{x_{1},y_{1},x_{2},y_{2},x_{3},y_{3}\in \mathbb{Z}^{d}}\left\vert
\left\langle \mathfrak{e}_{x_{1}},K_{\mathcal{Z},\mathcal{Z}^{(\tau
)}}^{(\omega ,\mathcal{E})}\mathfrak{e}_{y_{1}}\right\rangle _{\mathfrak{h}%
}\right\vert \,\left\vert \left\langle \mathfrak{e}_{x_{2}},K_{\mathcal{Z},%
\mathcal{Z}^{(\tau )}}^{(\omega ,\mathcal{E})}\mathfrak{e}%
_{y_{2}}\right\rangle _{\mathfrak{h}}\right\vert \,\left\vert \left\langle 
\mathfrak{e}_{x_{3}},K_{\mathcal{Z},\mathcal{Z}^{(\tau )}}^{(\omega ,%
\mathcal{E})}\mathfrak{e}_{y_{3}}\right\rangle _{\mathfrak{h}}\right\vert \\
\left\vert \varpi _{s}\left( a^{\ast }\left( \mathfrak{e}_{x_{1}}\right)
a\left( \mathfrak{e}_{y_{3}}\right) \right) \right\vert \,|\varpi _{s}(a(%
\mathfrak{e}_{y_{1}})a^{\ast }\left( \mathfrak{e}_{x_{2}}\right)
)|\,\left\vert \varpi _{s}\left( a(\mathfrak{e}_{y_{2}})a^{\ast }\left( 
\mathfrak{e}_{x_{3}}\right) \right) \right\vert \\
\leq \sup_{x_{3},y_{3}\in \mathbb{Z}^{d}}\left\vert \left\langle \mathfrak{e}%
_{x_{3}},K_{\mathcal{Z},\mathcal{Z}^{(\tau )}}^{(\omega ,\mathcal{E})}%
\mathfrak{e}_{y_{3}}\right\rangle _{\mathfrak{h}}\right\vert \sup_{x_{2}\in 
\mathbb{Z}^{d}}\sum_{y_{2}\in \mathbb{Z}^{d}}\left\vert \left\langle 
\mathfrak{e}_{x_{2}},K_{\mathcal{Z},\mathcal{Z}^{(\tau )}}^{(\omega ,%
\mathcal{E})}\mathfrak{e}_{y_{2}}\right\rangle _{\mathfrak{h}}\right\vert 
\frac{1}{|\cup \mathcal{Z}|}\sum_{x_{1},y_{1}\in \mathbb{Z}^{d}}\left\vert
\left\langle \mathfrak{e}_{x_{1}},K_{\mathcal{Z},\mathcal{Z}^{(\tau
)}}^{(\omega ,\mathcal{E})}\mathfrak{e}_{y_{1}}\right\rangle _{\mathfrak{h}%
}\right\vert \\
\sup_{x_{1}\in \mathbb{Z}^{d}}\sum_{y_{3}\in \mathbb{Z}^{d}}\left\vert
\varpi _{s}\left( a^{\ast }\left( \mathfrak{e}_{x_{1}}\right) a\left( 
\mathfrak{e}_{y_{3}}\right) \right) \right\vert \sup_{y_{1}\in \mathbb{Z}%
^{d}}\sum_{x_{2}\in \mathbb{Z}^{d}}\left\vert \varpi _{s}\left( a\left( 
\mathfrak{e}_{y_{1}}\right) a^{\ast }\left( \mathfrak{e}_{x_{2}}\right)
\right) \right\vert \sup_{y_{2}\in \mathbb{Z}^{d}}\sum_{x_{3}\in \mathbb{Z}%
^{d}}\left\vert \varpi _{s}\left( a\left( \mathfrak{e}_{y_{2}}\right)
a^{\ast }\left( \mathfrak{e}_{x_{3}}\right) \right) \right\vert .
\end{multline*}%
We can finally use Equations (\ref{fghfh})-(\ref{corollaire utile}) and (\ref%
{cool2}) to arrive from the last upper bound at 
\begin{equation*}
\sup_{\substack{ \beta \in \left( 0,\beta _{1}\right] ,\ \vartheta \in \left[
0,\vartheta _{1}\right] ,\ \lambda \in \left[ 0,\lambda _{1}\right]  \\ %
\omega \in \Omega ,\ s\in \left[ -s_{1},s_{1}\right] ,\ \mathcal{Z},\mathcal{%
Z}^{(\varrho )},\mathcal{Z}^{(\tau )}\in \mathfrak{Z}_{\text{f}}}}|\mathbf{K}%
_{1}|<\infty .
\end{equation*}%
The absolute value $|\mathbf{K}_{2}|$ of the other term of $\partial _{s}^{3}%
\mathrm{J}_{\mathcal{Z},\mathcal{Z}^{(\varrho )},\mathcal{Z}^{(\tau
)}}^{(\omega ,s\mathcal{E})}$ (see (\ref{K0})-(\ref{K2})) can be bounded
exactly in the same way. By the triangle inequality applied to (\ref{K0}),
this concludes the proof.
\end{proof}

We can now sharpen the result given in \cite[Corollary 4.20]{LDP}, stating
that the mapping $s\mapsto \mathrm{J}^{(s\mathcal{E})}$ defined by (\ref%
{generating_function}) is continuously differentiable with%
\begin{equation*}
\partial _{s}\mathrm{J}^{(s\mathcal{E})}=\lim_{L\rightarrow \infty }\frac{%
\varrho ^{(\omega )}\left( \mathbb{I}_{\Lambda _{L}}^{(\omega ,\mathcal{E})}%
\mathrm{e}^{s\left\vert \Lambda _{L}\right\vert \mathbb{I}_{\Lambda
_{L}}^{(\omega ,\mathcal{E})}}\right) }{\varrho ^{(\omega )}\left( \mathrm{e}%
^{s\left\vert \Lambda _{L}\right\vert \mathbb{I}_{\Lambda _{L}}^{(\omega ,%
\mathcal{E})}}\right) }.
\end{equation*}%
Thanks to Equation (\ref{proposition 4.9 new}) and Proposition \ref%
{3eme_derivee}, we now obtain the following assertion:

\begin{corollary}[Differentiability of generating functions]
\label{final corrollary}\mbox{
}\newline
There is a measurable subset $\tilde{\Omega}\subseteq \Omega $ of full
measure such that, for all $\beta \in \mathbb{R}^{+}$, $\vartheta ,\lambda
\in \mathbb{R}_{0}^{+}$, $\omega \in \tilde{\Omega}$, $\mathcal{E}\in
C_{0}^{0}(\mathbb{R};\mathbb{R}^{d})$ and $\vec{w}\in {\ \mathbb{R}}^{d}$
with $\left\Vert \vec{w}\right\Vert _{\mathbb{R}^{d}}=1$, the mapping $%
s\mapsto \mathrm{J}^{(s\mathcal{E})}$ from $\mathbb{R}$ to itself belongs to 
$C^{2}\left( \mathbb{R};\mathbb{R}\right) $ and 
\begin{eqnarray*}
\partial _{s}\mathrm{J}^{(s\mathcal{E})}|_{s=0} &=&x^{(\mathcal{E})}\doteq
\lim_{L\rightarrow \infty }\mathbb{E}\left[ \varrho ^{(\cdot )}\left( 
\mathbb{I}_{\Lambda _{L}}^{(\cdot ,\mathcal{E})}\right) \right]
=\lim_{L\rightarrow \infty }\varrho ^{(\omega )}\left( \mathbb{I}_{\Lambda
_{L}}^{(\omega ,\mathcal{E})}\right) \\
\partial _{s}^{2}\mathrm{J}^{(s\mathcal{E})}|_{s=0} &=&\lim_{L\rightarrow
\infty }\mathbb{E}\left[ \mathbf{F}_{L}^{(\cdot ,\mathcal{E})}\right]
=\lim_{L\rightarrow \infty }\mathbf{F}_{L}^{(\omega ,\mathcal{E})}\geq 0\ ,
\end{eqnarray*}%
where $\mathbf{F}_{L}^{(\omega ,\mathcal{E})}$ is the quantum fluctuation of
the linear response current defined by (\ref{qq}) for any $L\in \mathbb{R}%
_{0}^{+}$. See also (\ref{def currents0}) for the definition of the
macroscopic current density $x^{(\mathcal{E})}$.
\end{corollary}

\begin{proof}
\cite[Corollary 4.19]{LDP} states, among other things, the existence of a
measurable set $\tilde{\Omega}$ of full measure such that, for all $\beta
\in \mathbb{R}^{+}$, $\vartheta ,\lambda \in \mathbb{R}_{0}^{+}$, $\omega
\in \tilde{\Omega}$, $\mathcal{E}\in C_{0}^{0}(\mathbb{R};\mathbb{R}^{d})$, $%
\vec{w}\in {\ \mathbb{R}}^{d}$ with $\left\Vert \vec{w}\right\Vert _{\mathbb{%
R}^{d}}=1$ and $s\in \mathbb{R}$, 
\begin{equation}
\mathrm{J}^{(s\mathcal{E})}=\lim_{L_{\tau }\geq L_{\varrho }\geq
L\rightarrow \infty }\mathrm{J}_{\mathcal{\{}\Lambda _{L}\mathcal{\}},%
\mathcal{\{}\Lambda _{L_{\varrho }}\mathcal{\}},\mathcal{\{}\Lambda
_{L_{\tau }}\mathcal{\}}}^{(\omega ,s\mathcal{E})}.
\label{dddddddddvvvvvvvv}
\end{equation}%
Fix from now all parameters $\beta \in \mathbb{R}^{+}$, $\vartheta ,\lambda
\in \mathbb{R}_{0}^{+}$, $\omega \in \tilde{\Omega}$, $\mathcal{E}\in
C_{0}^{0}(\mathbb{R};\mathbb{R}^{d})$ and $\vec{w}\in {\ \mathbb{R}}^{d}$
with $\left\Vert \vec{w}\right\Vert _{\mathbb{R}^{d}}=1$. By combining
Equation (\ref{proposition 4.9 new}) and Proposition \ref{3eme_derivee} with
the mean value theorem and the (Arzel\`{a}-) Ascoli theorem \cite[Theorem A5]%
{Rudin}, there are three sequences 
\begin{equation}
\{L_{\tau }^{(n)}\}_{n\in \mathbb{N}},\{L_{\varrho }^{(n)}\}_{n\in \mathbb{N}%
},\{L^{(n)}\}_{n\in \mathbb{N}}\subseteq \mathbb{R}_{0}^{+},
\label{sequences}
\end{equation}%
with $L_{\tau }^{(n)}\geq L_{\varrho }^{(n)}\geq L^{(n)}$, such that, as $%
n\rightarrow \infty $, the mappings 
\begin{equation*}
s\mapsto \mathrm{J}_{\mathcal{\{}\Lambda _{L^{(n)}}\mathcal{\}},\mathcal{\{}%
\Lambda _{L_{\varrho }^{(n)}}\mathcal{\}},\mathcal{\{}\Lambda _{L_{\tau
}^{(n)}}\mathcal{\}}}^{(\omega ,s\mathcal{E})},\quad s\mapsto \partial _{s}%
\mathrm{J}_{\mathcal{\{}\Lambda _{L^{(n)}}\mathcal{\}},\mathcal{\{}\Lambda
_{L_{\varrho }^{(n)}}\mathcal{\}},\mathcal{\{}\Lambda _{L_{\tau }^{(n)}}%
\mathcal{\}}}^{(\omega ,s\mathcal{E})}\quad \text{and}\quad s\mapsto
\partial _{s}^{2}\mathrm{J}_{\mathcal{\{}\Lambda _{L^{(n)}}\mathcal{\}},%
\mathcal{\{}\Lambda _{L_{\varrho }^{(n)}}\mathcal{\}},\mathcal{\{}\Lambda
_{L_{\tau }^{(n)}}\mathcal{\}}}^{(\omega ,s\mathcal{E})}
\end{equation*}%
from $\mathbb{R}$ to itself converge uniformly for $s$ in any compact subset
of $\mathbb{R}$. So, the mapping $s\mapsto \mathrm{J}^{(s\mathcal{E})}$ from 
${\mathbb{R}}$ to itself is a $C^{2}$-function with 
\begin{equation*}
\partial _{s}\mathrm{J}^{(s\mathcal{E})}=\lim_{L_{\tau }\geq L_{\varrho
}\geq L\rightarrow \infty }\partial _{s}\mathrm{J}_{\mathcal{\{}\Lambda _{L}%
\mathcal{\}},\mathcal{\{}\Lambda _{L_{\varrho }}\mathcal{\}},\mathcal{\{}%
\Lambda _{L_{\tau }}\mathcal{\}}}^{(\omega ,s\mathcal{E})}=\lim_{L%
\rightarrow \infty }\left( \frac{\varrho ^{(\omega )}\left( \mathbb{I}%
_{\Lambda _{L}}^{(\omega ,\mathcal{E})}\mathrm{e}^{s\left\vert \Lambda
_{L}\right\vert \mathbb{I}_{\Lambda _{L}}^{(\omega ,\mathcal{E})}}\right) }{%
\varrho ^{(\omega )}\left( \mathrm{e}^{s\left\vert \Lambda _{L}\right\vert 
\mathbb{I}_{\Lambda _{L}}^{(\omega ,\mathcal{E})}}\right) }\right)
\end{equation*}%
and 
\begin{eqnarray*}
\partial _{s}^{2}\mathrm{J}^{(s\mathcal{E})} &=&\lim_{L_{\tau }\geq
L_{\varrho }\geq L\rightarrow \infty }\partial _{s}^{2}\mathrm{J}_{\mathcal{%
\{}\Lambda _{L}\mathcal{\}},\mathcal{\{}\Lambda _{L_{\varrho }}\mathcal{\}},%
\mathcal{\{}\Lambda _{L_{\tau }}\mathcal{\}}}^{(\omega ,s\mathcal{E})} \\
&=&\lim_{L\rightarrow \infty }|\Lambda _{L}|\left( \frac{\varrho ^{(\omega
)}\left( \left( \mathbb{I}_{\Lambda _{L}}^{(\omega ,\mathcal{E})}\right) ^{2}%
\mathrm{e}^{s\left\vert \Lambda _{L}\right\vert \mathbb{I}_{\Lambda
_{L}}^{(\omega ,\mathcal{E})}}\right) \varrho ^{(\omega )}\left( \mathrm{e}%
^{s\left\vert \Lambda _{L}\right\vert \mathbb{I}_{\Lambda _{L}}^{(\omega ,%
\mathcal{E})}}\right) -\left( \varrho ^{(\omega )}\left( \mathbb{I}_{\Lambda
_{L}}^{(\omega ,\mathcal{E})}\mathrm{e}^{s\left\vert \Lambda _{L}\right\vert 
\mathbb{I}_{\Lambda _{L}}^{(\omega ,\mathcal{E})}}\right) \right) ^{2}}{%
\left( \varrho ^{(\omega )}\left( \mathrm{e}^{s\left\vert \Lambda
_{L}\right\vert \mathbb{I}_{\Lambda _{L}}^{(\omega ,\mathcal{E})}}\right)
\right) ^{2}}\right) .
\end{eqnarray*}%
See (\ref{aa}). Note that the above limits for the first- and second-order
derivatives do not need to be taken only along subsequences, by the (Arzel%
\`{a}-) Ascoli theorem \cite[Theorem A5]{Rudin} and (\ref{dddddddddvvvvvvvv}%
). In particular, for $s=0$, 
\begin{equation}
\partial _{s}\mathrm{J}^{(s\mathcal{E})}|_{s=0}=\lim_{L\rightarrow \infty }%
\mathbb{E}\left[ \varrho ^{(\cdot )}\left( \mathbb{I}_{\Lambda _{L}}^{(\cdot
,\mathcal{E})}\right) \right] =\lim_{L\rightarrow \infty }\varrho ^{(\omega
)}\left( \mathbb{I}_{\Lambda _{L}}^{(\omega ,\mathcal{E})}\right)
\label{current utiles}
\end{equation}%
and%
\begin{eqnarray}
\partial _{s}^{2}\mathrm{J}^{(s\mathcal{E})}|_{s=0} &=&\lim_{L\rightarrow
\infty }|\Lambda _{L}|\mathbb{E}\left[ \varrho ^{(\cdot )}\left( \left( 
\mathbb{I}_{\Lambda _{L}}^{(\cdot ,\mathcal{E})}\right) ^{2}\right) -\left(
\varrho ^{(\cdot )}\left( \mathbb{I}_{\Lambda _{L}}^{(\cdot ,\mathcal{E}%
)}\right) \right) ^{2}\right] \text{ }  \notag \\
&=&\lim_{L\rightarrow \infty }|\Lambda _{L}|\left( \varrho ^{(\omega
)}\left( \left( \mathbb{I}_{\Lambda _{L}}^{(\omega ,\mathcal{E})}\right)
^{2}\right) -\left( \varrho ^{(\omega )}\left( \mathbb{I}_{\Lambda
_{L}}^{(\omega ,\mathcal{E})}\right) \right) ^{2}\right) .
\label{current utiles_2}
\end{eqnarray}%
By Equations (\ref{qq})-(\ref{qq2}) and (\ref{current utiles_2}), $\partial
_{s}^{2}\mathrm{J}^{(s\mathcal{E})}|_{s=0}$ is the thermodynamic limit of
the quantum fluctuations of linear response currents.
\end{proof}

From the proof of Proposition \ref{3eme_derivee}, it is apparent that the $n$%
-th derivative $\partial _{s}^{n}\mathrm{J}_{\mathcal{Z},\mathcal{Z}%
^{(\varrho )},\mathcal{Z}^{(\tau )}}^{(\omega ,s\mathcal{E})}$, $n\in 
\mathbb{N}$, has the following structure:%
\begin{eqnarray*}
\partial _{s}^{n}\mathrm{J}_{\mathcal{Z},\mathcal{Z}^{(\varrho )},\mathcal{Z}%
^{(\tau )}}^{(\omega ,s\mathcal{E})} &=&\frac{1}{|\cup \mathcal{Z}|}%
\sum\limits_{k=1}^{n}\sum\limits_{x_{k},y_{k}\in \mathbb{Z}^{d}}\left\langle 
\mathfrak{e}_{x_{1}},K_{\mathcal{Z},\mathcal{Z}^{(\tau )}}^{(\omega ,%
\mathcal{E})}\mathfrak{e}_{y_{1}}\right\rangle _{\mathfrak{h}}\cdots
\left\langle \mathfrak{e}_{x_{k}},K_{\mathcal{Z},\mathcal{Z}^{(\tau
)}}^{(\omega ,\mathcal{E})}\mathfrak{e}_{y_{k}}\right\rangle _{\mathfrak{h}}
\\
&&\times \varpi _{s}^{T}(a^{\ast }\left( \mathfrak{e}_{x_{1}}\right) a\left( 
\mathfrak{e}_{y_{1}}\right) ;\cdots ;a^{\ast }\left( \mathfrak{e}%
_{x_{n}}\right) a\left( \mathfrak{e}_{x_{n}}\right) ) \\
&=&\frac{1}{|\cup \mathcal{Z}|}\sum\limits_{g\in \mathcal{G}%
_{n}^{c}}\sum\limits_{k=1}^{n}\sum\limits_{x_{k},y_{k}\in \mathbb{Z}^{d}}%
\mathrm{sign}(g)\left\langle \mathfrak{e}_{x_{1}},K_{\mathcal{Z},\mathcal{Z}%
^{(\tau )}}^{(\omega ,\mathcal{E})}\mathfrak{e}_{y_{1}}\right\rangle _{%
\mathfrak{h}}\cdots \left\langle \mathfrak{e}_{x_{k}},K_{\mathcal{Z},%
\mathcal{Z}^{(\tau )}}^{(\omega ,\mathcal{E})}\mathfrak{e}%
_{y_{k}}\right\rangle _{\mathfrak{h}} \\
&&\qquad \qquad \qquad \qquad \qquad \qquad \qquad \qquad \times
\prod\limits_{l\in g}k_{s}\left( l;x_{1},y_{1},\ldots ,x_{n},y_{n}\right) 
\text{ },
\end{eqnarray*}%
where $\mathcal{G}_{n}^{c}$ is the set of all connected oriented graphs $g$
such that, for each vertex $v\in \{1,\ldots ,n\}$ of $g\in \mathcal{G}%
_{n}^{c}$, there is exactly one line of the form $(v,\tilde{v}_{1})\in g$
and exactly one line of the form $(\tilde{v}_{2},v)\in g$, for some $\tilde{v%
}_{1},\tilde{v}_{2}\in \{1,\ldots ,n\}$. The constants $k_{s}\left(
l;x_{1},y_{1},\ldots ,x_{n},y_{n}\right) $, $l\in \{1,\ldots ,n\}^{2}$, $%
x_{1},y_{1},\ldots ,x_{n},y_{n}\in \mathbb{Z}^{d}$, are defined by%
\begin{equation*}
k_{s}\left( (i,j);x_{1},y_{1},\ldots ,x_{n},y_{n}\right) \doteq \left\{ 
\begin{array}{ccc}
\varpi _{s}(a^{\ast }\left( \mathfrak{e}_{x_{i}}\right) a\left( \mathfrak{e}%
_{y_{j}}\right) ) & \text{if} & i\leq j. \\ 
\varpi _{s}(a\left( \mathfrak{e}_{y_{j}}\right) a^{\ast }\left( \mathfrak{e}%
_{x_{i}}\right) ) & \text{if} & i>j.%
\end{array}%
\right.
\end{equation*}%
The quantity $\mathrm{sign}(g)\in \{-1,1\}$ is a sign only depending on the
graph $g\in \mathcal{G}_{n}^{c}$. By using this expression, exactly as in
the special case $n=3$, for any fixed $n\in \mathbb{N}$ and electric field $%
\mathcal{E}$ one can bound the $n$-th derivative $\partial _{s}^{n}\mathrm{J}%
_{\mathcal{Z},\mathcal{Z}^{(\varrho )},\mathcal{Z}^{(\tau )}}^{(\omega ,s%
\mathcal{E})}$ uniformly. This implies that the generating function $%
s\mapsto \mathrm{J}^{(s\mathcal{E})}$ defined by (\ref{generating_function})
is a \emph{smooth} function of $s\in \mathbb{R}$, by the (Arzel\`{a}-)
Ascoli theorem \cite[Theorem A5]{Rudin} used as in the proof of Corollary %
\ref{final corrollary}. We refrain from working out the full arguments to
prove this claim since absolutely no new conceptual ingredient would appear
in this generalization.

\subsection{Non-Vanishing Second Derivative of Generating Functions at the
Origin\label{positivite}}

We discuss necessary conditions for 
\begin{equation}
\partial _{s}^{2}\mathrm{J}^{(s\mathcal{E})}|_{s=0}\neq 0,
\label{blabla_2blabla_2}
\end{equation}%
which is a condition appearing in Theorem \ref{Quantum fluctuations and rate
function} (ii). In other words, the aim of this section is to prove Theorem %
\ref{new theorem}. To this end, it is convenient to write this quantity by
means of the one-particle Hilbert space $\mathfrak{h}$.

\begin{lemma}[Quantum fluctuations on the one-particle Hilbert space]
\label{equation inmportantelemma}\mbox{}\newline
For all $\beta \in \mathbb{R}^{+}$, $\vartheta ,\lambda \in \mathbb{R}%
_{0}^{+}$, $\mathcal{E}\in C_{0}^{0}(\mathbb{R};\mathbb{R}^{d})$ and $\vec{w}%
\in {\ \mathbb{R}}^{d}$ with $\left\Vert \vec{w}\right\Vert _{\mathbb{R}%
^{d}}=1$,%
\begin{equation*}
\partial _{s}^{2}\mathrm{J}^{(s\mathcal{E})}|_{s=0}=\lim_{L\rightarrow
\infty }\frac{1}{|\Lambda _{L}|}\mathbb{E}\left[ \mathrm{Tr}_{\mathfrak{h}%
}\left( K_{\{\Lambda _{L}\},\{\mathbb{Z}^{d}\}}^{(\cdot ,\mathcal{E})}\frac{1%
}{1+\mathrm{e}^{-\beta h^{(\cdot )}}}K_{\{\Lambda _{L}\},\{\mathbb{Z}%
^{d}\}}^{(\cdot ,\mathcal{E})}\frac{1}{1+\mathrm{e}^{\beta h^{(\cdot )}}}%
\right) \right]
\end{equation*}%
with $\mathrm{Tr}_{\mathfrak{h}}$ being the trace on $\mathfrak{h}\doteq
\ell ^{2}(\mathbb{Z}^{d};\mathbb{C})$.
\end{lemma}

\begin{proof}
Fix all parameters of the lemma. Using Equations (\ref{aa}) and (\ref%
{bilinear idiot}) together with the quasi-free property of $\varrho
^{(\omega )}$, one obtains from (\ref{current utiles_2}) that%
\begin{multline*}
\partial _{s}^{2}\mathrm{J}^{(s\mathcal{E})}|_{s=0}=\lim_{L\rightarrow
\infty }\frac{1}{|\Lambda _{L}|}\sum\limits_{x,y,u,v\in \mathbb{Z}%
^{d}}\left\langle \mathfrak{e}_{x},K_{\{\Lambda _{L}\},\{\mathbb{Z}%
^{d}\}}^{(\omega ,\mathcal{E})}\mathfrak{e}_{y}\right\rangle _{\mathfrak{h}%
}\left\langle \mathfrak{e}_{u},K_{\{\Lambda _{L}\},\{\mathbb{Z}%
^{d}\}}^{(\omega ,\mathcal{E})}\mathfrak{e}_{v}\right\rangle _{\mathfrak{h}}
\\
\times \varrho ^{(\omega )}\left( a\left( \mathfrak{e}_{y}\right) a\left( 
\mathfrak{e}_{u}\right) ^{\ast }\right) \varrho ^{(\omega )}\left( a\left( 
\mathfrak{e}_{x}\right) ^{\ast }a\left( \mathfrak{e}_{v}\right) \right) ,
\end{multline*}
because of the identity%
\begin{equation*}
\rho \left( a(\mathfrak{e}_{x})^{\ast }a(\mathfrak{e}_{y})a(\mathfrak{e}%
_{u})^{\ast }a(\mathfrak{e}_{v})\right) =\rho \left( a(\mathfrak{e}%
_{x})^{\ast }a(\mathfrak{e}_{y})\right) \rho \left( a(\mathfrak{e}%
_{u})^{\ast }a(\mathfrak{e}_{v})\right) +\rho \left( a(\mathfrak{e}_{y})a(%
\mathfrak{e}_{u})^{\ast }\right) \rho \left( a(\mathfrak{e}_{x})^{\ast }a(%
\mathfrak{e}_{v})\right)
\end{equation*}%
for any $x,y,u,v\in \mathbb{Z}^{d}$ and quasi-free state $\rho $ on $%
\mathcal{U}$, see (\ref{CAR}) and (\ref{ass O0-00bis}). By Equation (\ref%
{2-point correlation function}) and straightforward computations, the
assertion follows.
\end{proof}

Therefore, (\ref{blabla_2blabla_2}) holds true if 
\begin{equation*}
\lim_{L\rightarrow \infty }\left\{ \frac{1}{|\Lambda _{L}|}\left\vert 
\mathrm{Tr}_{\mathfrak{h}}\left( K_{\{\Lambda _{L}\},\{\mathbb{Z}%
^{d}\}}^{(\omega ,\mathcal{E})}\frac{1}{1+\mathrm{e}^{-\beta h^{(\omega )}}}%
K_{\{\Lambda _{L}\},\{\mathbb{Z}^{d}\}}^{(\omega ,\mathcal{E})}\frac{1}{1+%
\mathrm{e}^{\beta h^{(\omega )}}}\right) \right\vert \right\} \geq
\varepsilon >0
\end{equation*}%
for some strictly positive constant $\varepsilon \in \mathbb{R}^{+}$. In
order to verify this bound, we start with an elementary observation:

\begin{lemma}[Quantum fluctuations and the Hilbert-Schmidt norm of $%
K_{\{\Lambda _{L}\},\{\mathbb{Z}^{d}\}}^{(\protect\omega ,\mathcal{E})}$]
\label{long_lemme copy(2)}\mbox{}\newline
For all $\beta \in \mathbb{R}^{+}$, $\vartheta ,\lambda \in \mathbb{R}%
_{0}^{+}$, $\omega \in \Omega $, $\mathcal{E}\in C_{0}^{0}(\mathbb{R};%
\mathbb{R}^{d})$ and $\vec{w}\in {\mathbb{R}}^{d}$ with $\left\Vert \vec{w}%
\right\Vert _{\mathbb{R}^{d}}=1$, 
\begin{multline*}
\mathrm{Tr}_{\mathfrak{h}}\left( K_{\{\Lambda _{L}\},\{\mathbb{Z}%
^{d}\}}^{(\omega ,\mathcal{E})}\frac{1}{1+\mathrm{e}^{-\beta h^{(\omega )}}}%
K_{\{\Lambda _{L}\},\{\mathbb{Z}^{d}\}}^{(\omega ,\mathcal{E})}\frac{1}{1+%
\mathrm{e}^{\beta h^{(\omega )}}}\right) \\
\geq \frac{1}{\left( 1+\mathrm{e}^{\beta \left( 2d\left( 2+\vartheta \right)
+\lambda \right) }\right) ^{2}}\mathrm{Tr}_{\mathfrak{h}}\left( \left(
K_{\{\Lambda _{L}\},\{\mathbb{Z}^{d}\}}^{(\omega ,\mathcal{E})}\right)
^{\ast }K_{\{\Lambda _{L}\},\{\mathbb{Z}^{d}\}}^{(\omega ,\mathcal{E}%
)}\right) .
\end{multline*}
\end{lemma}

\begin{proof}
Fix all parameters of the lemma. By the functional calculus, $(1+\mathrm{e}%
^{\pm \beta h^{(\omega )}})^{-1}\ $are positive operators satisfying 
\begin{equation*}
\frac{1}{1+\mathrm{e}^{\pm \beta h^{(\omega )}}}\geq \frac{1}{1+\mathrm{e}%
^{\beta \sup_{\omega \in \Omega }\Vert h^{(\omega )}\Vert _{\mathcal{B}%
\left( \mathfrak{h}\right) }}}\mathbf{1}_{\mathfrak{h}},
\end{equation*}%
while, for any $\omega =\left( \omega _{1},\omega _{2}\right) \in \Omega $
and $\lambda ,\vartheta \in \mathbb{R}_{0}^{+}$,%
\begin{equation}
\Vert h^{(\omega )}\Vert _{\mathcal{B}\left( \mathfrak{h}\right) }\leq \Vert
\Delta _{\omega ,\vartheta }\Vert _{\mathcal{B}\left( \mathfrak{h}\right)
}+\lambda \Vert \omega _{1}\Vert _{\mathcal{B}\left( \mathfrak{h}\right)
}\leq 2d\left( 2+\vartheta \right) +\lambda ,  \label{borne trivial}
\end{equation}%
see (\ref{equation sup})-(\ref{eq:Ham_lap_pot}). Since $K_{\left\{ \Lambda
_{L}\right\} ,\left\{ \mathbb{Z}^{d}\right\} }^{(\omega ,\mathcal{E})}$ is a
self-adjoint operator (see (\ref{definition K}) or (\ref{definition
Kbisdefinition Kbis}) below), it thus suffices to use the cyclicity of the
trace to prove the lemma.
\end{proof}

Recall that $K_{\left\{ \Lambda _{L}\right\} ,\left\{ \mathbb{Z}^{d}\right\}
}^{(\omega ,\mathcal{E})}$ is defined by (\ref{definition K}), that is in
this case, 
\begin{equation}
K_{\{\Lambda _{L}\},\left\{ \mathbb{Z}^{d}\right\} }^{(\omega ,\mathcal{E}%
)}\doteq \underset{k,q=1}{\sum^{d}}w_{k}\int_{-\infty }^{0}\left\{ \mathcal{E%
}\left( \alpha \right) \right\} _{q}\left( \delta _{k,q}\mathbf{M}%
_{k}^{(L,\omega )}+\int\nolimits_{0}^{-\alpha }\mathbf{N}_{\gamma
,q,k}^{(L,\omega )}\mathrm{d}\gamma \right) \mathrm{d}\alpha ,
\label{definition Kbisdefinition Kbis}
\end{equation}%
where, for any $k,q\in \{1,\ldots ,d\}$, $\gamma \in \mathbb{R}$, $\vartheta
,\lambda \in \mathbb{R}_{0}^{+}$, $\omega \in \Omega $ and $L\in \mathbb{R}%
^{+}$, 
\begin{eqnarray}
\mathbf{M}_{k}^{(L,\omega )} &\doteq &\underset{x,x+e_{k}\in \Lambda _{L}}{%
\sum }2\Re \mathrm{e}\{S_{x+e_{k},x}^{(\omega )}\}  \label{M} \\
\mathbf{N}_{\gamma ,q,k}^{(L,\omega )} &\doteq &\underset{%
x,y,x+e_{k},y+e_{q}\in \Lambda _{L}}{\sum }4i\left[ \mathrm{e}^{-i\gamma
h^{(\omega )}}\Im \mathrm{m}\{S_{y+e_{q},y}^{(\omega )}\}\mathrm{e}^{i\gamma
h^{(\omega )}},\Im \mathrm{m}\{S_{x+e_{k},x}^{(\omega )}\}\right]  \label{N}
\end{eqnarray}%
with $S_{x,y}^{(\omega )}$ being the single-hopping operators defined by (%
\ref{shift})-(\ref{shift2}) for any $x,y\in \mathbb{Z}^{d}$.

The square of the Hilbert-Schmidt norm of $K_{\{\Lambda _{L}\},\{\mathbb{Z}%
^{d}\}}^{(\omega ,\mathcal{E})}$ is obviously equal to%
\begin{equation*}
\mathrm{Tr}_{\mathfrak{h}}\left( \left( K_{\{\Lambda _{L}\},\{\mathbb{Z}%
^{d}\}}^{(\omega ,\mathcal{E})}\right) ^{\ast }K_{\{\Lambda _{L}\},\{\mathbb{%
Z}^{d}\}}^{(\omega ,\mathcal{E})}\right) =\sum_{z\in \mathbb{Z}%
^{d}}\left\Vert K_{\{\Lambda _{L}\},\{\mathbb{Z}^{d}\}}^{(\omega ,\mathcal{E}%
)}\mathfrak{e}_{z}\right\Vert _{\mathfrak{h}}^{2}
\end{equation*}%
and, consequently, we derive an explicit expression for the vectors 
\begin{equation*}
K_{\{\Lambda _{L}\},\{\mathbb{Z}^{d}\}}^{(\omega ,\mathcal{E})}\mathfrak{e}%
_{z}\in \mathfrak{h},\qquad z\in \mathbb{Z}^{d}.
\end{equation*}%
This can be directly obtained from Equation (\ref{definition Kbisdefinition
Kbis}) together with the following assertion:

\begin{lemma}[Explicit computations of $\mathbf{M}_{k}^{(L,\protect\omega )}$%
\textbf{\ and }$\mathbf{N}_{\protect\gamma ,q,k}^{(L,\protect\omega )}$ in
the canonical basis]
\label{long_lemme copy(3)}\mbox{}\newline
For all $k,q\in \{1,\ldots ,d\}$, $\gamma \in \mathbb{R}$, $\vartheta
,\lambda \in \mathbb{R}_{0}^{+}$, $\omega \in \Omega $, $\gamma \in \mathbb{R%
}$, $L\geq 2$ and $z\in \Lambda _{L/2}$,%
\begin{equation*}
\mathbf{M}_{k}^{(L,\omega )}\mathfrak{e}_{z}=\langle \mathfrak{e}%
_{z-e_{k}},\Delta _{\omega ,\vartheta }\mathfrak{e}_{z}\rangle _{\mathfrak{h}%
}\mathfrak{e}_{z-e_{k}}+\langle \mathfrak{e}_{z+e_{k}},\Delta _{\omega
,\vartheta }\mathfrak{e}_{z}\rangle _{\mathfrak{h}}\mathfrak{e}_{z+e_{k}}
\end{equation*}%
and, in the limit $L\rightarrow \infty $, 
\begin{equation*}
\mathbf{N}_{\gamma ,q,k}^{(L,\omega )}\mathfrak{e}_{z}=\sum_{x,y\in \mathbb{Z%
}^{d}}\zeta _{x,y,z}\mathfrak{e}_{x}+\mathbf{R}_{\gamma ,q,k}^{(L,\omega )}%
\mathfrak{e}_{z}\ ,\qquad \sum_{x,y\in \mathbb{Z}^{d}}\left\vert \zeta
_{x,y,z}\right\vert ^{2}<\infty \ ,
\end{equation*}%
with $\mathbf{R}_{\gamma ,q,k}^{(L,\omega )}\in \mathcal{B}\left( \mathfrak{h%
}\right) $ satisfying%
\begin{equation}
\lim_{L\rightarrow \infty }\left\Vert \mathbf{R}_{\gamma ,q,k}^{(L,\omega
)}\right\Vert _{\mathcal{B}\left( \mathfrak{h}\right) }=0,  \label{limit}
\end{equation}%
uniformly with respect to $\omega \in \Omega $, $\lambda \in \mathbb{R}%
_{0}^{+}$ and $\vartheta $,$\gamma $ in compact subsets of $\mathbb{R}%
_{0}^{+}$ and $\mathbb{R}$, respectively, and where, for any $x,y,z\in 
\mathbb{Z}^{d}$,%
\begin{eqnarray*}
\zeta _{x,y,z} &\doteq &i(1+\vartheta \omega
_{2}(\{x-e_{k},x\}))(1+\vartheta \omega _{2}(\{y,y+e_{q}\}))\langle 
\mathfrak{e}_{x-e_{k}},\mathrm{e}^{-i\gamma h^{(\omega )}}\mathfrak{e}%
_{y+e_{q}}\rangle _{\mathfrak{h}}\langle \mathfrak{e}_{y},\mathrm{e}%
^{i\gamma h^{(\omega )}}\mathfrak{e}_{z}\rangle _{\mathfrak{h}} \\
&&-i(1+\vartheta \omega _{2}(\{x-e_{k},x\}))(1+\vartheta \overline{\omega
_{2}(\{y+e_{q},y\})})\langle \mathfrak{e}_{x-e_{k}},\mathrm{e}^{-i\gamma
h^{(\omega )}}\mathfrak{e}_{y}\rangle _{\mathfrak{h}}\langle \mathfrak{e}%
_{y+e_{q}},\mathrm{e}^{i\gamma h^{(\omega )}}\mathfrak{e}_{z}\rangle _{%
\mathfrak{h}} \\
&&-i(1+\vartheta \overline{\omega _{2}(\{x+e_{k},x\})})(1+\vartheta \omega
_{2}(\{y,y+e_{q}\}))\langle \mathfrak{e}_{x+e_{k}},\mathrm{e}^{-i\gamma
h^{(\omega )}}\mathfrak{e}_{y+e_{q}}\rangle _{\mathfrak{h}}\langle \mathfrak{%
e}_{y},\mathrm{e}^{i\gamma h^{(\omega )}}\mathfrak{e}_{z}\rangle _{\mathfrak{%
h}} \\
&&+i(1+\vartheta \overline{\omega _{2}(\{x+e_{k},x\})})(1+\vartheta 
\overline{\omega _{2}(\{y+e_{q},y\})})\langle \mathfrak{e}_{x+e_{k}},\mathrm{%
e}^{-i\gamma h^{(\omega )}}\mathfrak{e}_{y}\rangle _{\mathfrak{h}}\langle 
\mathfrak{e}_{y+e_{q}},\mathrm{e}^{i\gamma h^{(\omega )}}\mathfrak{e}%
_{z}\rangle _{\mathfrak{h}} \\
&&-i(1+\vartheta \omega _{2}(\{y,y+e_{q}\}))(1+\vartheta \omega
_{2}(\{z,z+e_{k}\}))\langle \mathfrak{e}_{y},\mathrm{e}^{i\gamma h^{(\omega
)}}\mathfrak{e}_{z+e_{k}}\rangle _{\mathfrak{h}}\langle \mathfrak{e}_{x},%
\mathrm{e}^{-i\gamma h^{(\omega )}}\mathfrak{e}_{y+e_{q}}\rangle _{\mathfrak{%
h}} \\
&&+i(1+\vartheta \omega _{2}(\{y,y+e_{q}\}))(1+\vartheta \overline{\omega
_{2}(\{z,z-e_{k}\})})\langle \mathfrak{e}_{y},\mathrm{e}^{i\gamma h^{(\omega
)}}\mathfrak{e}_{z-e_{k}}\rangle _{\mathfrak{h}}\langle \mathfrak{e}_{x},%
\mathrm{e}^{-i\gamma h^{(\omega )}}\mathfrak{e}_{y+e_{q}}\rangle _{\mathfrak{%
h}} \\
&&+i(1+\vartheta \overline{\omega _{2}(\{y+e_{q},y\})})(1+\vartheta \omega
_{2}(\{z,z+e_{k}\}))\langle \mathfrak{e}_{y+e_{q}},\mathrm{e}^{i\gamma
h^{(\omega )}}\mathfrak{e}_{z+e_{k}}\rangle _{\mathfrak{h}}\langle \mathfrak{%
e}_{x},\mathrm{e}^{-i\gamma h^{(\omega )}}\mathfrak{e}_{y}\rangle _{%
\mathfrak{h}} \\
&&-i(1+\vartheta \overline{\omega _{2}(\{y+e_{q},y\})})(1+\vartheta 
\overline{\omega _{2}(\{z,z-e_{k}\})})\langle \mathfrak{e}_{y+e_{q}},\mathrm{%
e}^{i\gamma h^{(\omega )}}\mathfrak{e}_{z-e_{k}}\rangle _{\mathfrak{h}%
}\langle \mathfrak{e}_{x},\mathrm{e}^{-i\gamma h^{(\omega )}}\mathfrak{e}%
_{y}\rangle _{\mathfrak{h}}.
\end{eqnarray*}
\end{lemma}

\begin{proof}
Fix in all the proof $k,q\in \{1,\ldots ,d\}$, $\vartheta ,\lambda \in 
\mathbb{R}_{0}^{+}$, $\omega \in \Omega $, $\gamma \in \mathbb{R}$, $L\geq 2$
and $z\in \Lambda _{L/2}$. Since, by (\ref{shift})-(\ref{shift2}), for any $%
x,y\in \mathbb{Z}^{d}$, 
\begin{equation*}
2\Re \mathrm{e}\{S_{x,y}^{(\omega )}\}=\langle \mathfrak{e}_{y},\Delta
_{\omega ,\vartheta }\mathfrak{e}_{x}\rangle _{\mathfrak{h}}P_{\left\{
y\right\} }s_{y-x}P_{\left\{ x\right\} }+\langle \mathfrak{e}_{x},\Delta
_{\omega ,\vartheta }\mathfrak{e}_{y}\rangle _{\mathfrak{h}}P_{\left\{
x\right\} }s_{x-y}P_{\left\{ y\right\} },
\end{equation*}%
we deduce from (\ref{M}) together with (\ref{orthogonal projection}) and (%
\ref{shift}) that 
\begin{align*}
\mathbf{M}_{k}^{(L,\omega )}\mathfrak{e}_{z}& =\underset{x,x+e_{k}\in
\Lambda _{L}}{\sum }\left( \delta _{z,x+e_{k}}\langle \mathfrak{e}%
_{x},\Delta _{\omega ,\vartheta }\mathfrak{e}_{x+e_{k}}\rangle _{\mathfrak{h}%
}\mathfrak{e}_{x}+\delta _{z,x}\langle \mathfrak{e}_{x+e_{k}},\Delta
_{\omega ,\vartheta }\mathfrak{e}_{x}\rangle _{\mathfrak{h}}\mathfrak{e}%
_{x+e_{k}}\right) \\
& =\mathbf{1}\left[ z\in \Lambda _{L}\right] \mathbf{1}\left[ \left(
z-e_{k}\right) \in \Lambda _{L}\right] \langle \mathfrak{e}_{z-e_{k}},\Delta
_{\omega ,\vartheta }\mathfrak{e}_{z}\rangle _{\mathfrak{h}}\mathfrak{e}%
_{z-e_{k}} \\
& \qquad \qquad \qquad \qquad +\mathbf{1}\left[ z\in \Lambda _{L}\right] 
\mathbf{1}\left[ (z+e_{k})\in \Lambda _{L}\right] \langle \mathfrak{e}%
_{z+e_{k}},\Delta _{\omega ,\vartheta }\mathfrak{e}_{z}\rangle _{\mathfrak{h}%
}\mathfrak{e}_{z+e_{k}}.
\end{align*}%
If $z\in \Lambda _{L/2}\subseteq \Lambda _{L}$ and $L\geq 2$ then,
obviously, $z,\left( z-e_{k}\right) ,\left( z+e_{k}\right) \in \Lambda _{L}$
and the last equality yields the first assertion.

Since, again by (\ref{shift})-(\ref{shift2}), for any $x,y\in \mathbb{Z}^{d}$%
, 
\begin{equation*}
2\Im \mathrm{m}\{S_{x,y}^{(\omega )}\}=i\left( \langle \mathfrak{e}%
_{y},\Delta _{\omega ,\vartheta }\mathfrak{e}_{x}\rangle _{\mathfrak{h}%
}P_{\left\{ y\right\} }s_{y-x}P_{\left\{ x\right\} }-\langle \mathfrak{e}%
_{x},\Delta _{\omega ,\vartheta }\mathfrak{e}_{y}\rangle _{\mathfrak{h}%
}P_{\left\{ x\right\} }s_{x-y}P_{\left\{ y\right\} }\right) ,
\end{equation*}%
we compute that, for any $x,y\in \mathbb{Z}^{d}$, 
\begin{eqnarray*}
&&4i\left[ \mathrm{e}^{-i\gamma h^{(\omega )}}\Im \mathrm{m}%
\{S_{y+e_{q},y}^{(\omega )}\}\mathrm{e}^{i\gamma h^{(\omega )}},\Im \mathrm{m%
}\{S_{x+e_{k},x}^{(\omega )}\}\right] \\
&=&i\langle \mathfrak{e}_{x+e_{k}},\Delta _{\omega ,\vartheta }\mathfrak{e}%
_{x}\rangle _{\mathfrak{h}}\langle \mathfrak{e}_{y+e_{q}},\Delta _{\omega
,\vartheta }\mathfrak{e}_{y}\rangle _{\mathfrak{h}}s_{e_{k}}P_{\left\{
x\right\} }\mathrm{e}^{-i\gamma h^{(\omega )}}s_{e_{q}}P_{\left\{ y\right\} }%
\mathrm{e}^{i\gamma h^{(\omega )}} \\
&&-i\langle \mathfrak{e}_{x+e_{k}},\Delta _{\omega ,\vartheta }\mathfrak{e}%
_{x}\rangle _{\mathfrak{h}}\langle \mathfrak{e}_{y},\Delta _{\omega
,\vartheta }\mathfrak{e}_{y+e_{q}}\rangle _{\mathfrak{h}}s_{e_{k}}P_{\left\{
x\right\} }\mathrm{e}^{-i\gamma h^{(\omega )}}s_{-e_{q}}P_{\left\{
y+e_{q}\right\} }\mathrm{e}^{i\gamma h^{(\omega )}} \\
&&-i\langle \mathfrak{e}_{x},\Delta _{\omega ,\vartheta }\mathfrak{e}%
_{x+e_{k}}\rangle _{\mathfrak{h}}\langle \mathfrak{e}_{y+e_{q}},\Delta
_{\omega ,\vartheta }\mathfrak{e}_{y}\rangle _{\mathfrak{h}%
}s_{-e_{k}}P_{\left\{ x+e_{k}\right\} }\mathrm{e}^{-i\gamma h^{(\omega
)}}s_{e_{q}}P_{\left\{ y\right\} }\mathrm{e}^{i\gamma h^{(\omega )}} \\
&&+i\langle \mathfrak{e}_{x},\Delta _{\omega ,\vartheta }\mathfrak{e}%
_{x+e_{k}}\rangle _{\mathfrak{h}}\langle \mathfrak{e}_{y},\Delta _{\omega
,\vartheta }\mathfrak{e}_{y+e_{q}}\rangle _{\mathfrak{h}}s_{-e_{k}}P_{\left%
\{ x+e_{k}\right\} }\mathrm{e}^{-i\gamma h^{(\omega )}}s_{-e_{q}}P_{\left\{
y+e_{q}\right\} }\mathrm{e}^{i\gamma h^{(\omega )}} \\
&&-i\langle \mathfrak{e}_{y+e_{q}},\Delta _{\omega ,\vartheta }\mathfrak{e}%
_{y}\rangle _{\mathfrak{h}}\langle \mathfrak{e}_{x+e_{k}},\Delta _{\omega
,\vartheta }\mathfrak{e}_{x}\rangle _{\mathfrak{h}}\mathrm{e}^{-i\gamma
h^{(\omega )}}s_{e_{q}}P_{\left\{ y\right\} }\mathrm{e}^{i\gamma h^{(\omega
)}}s_{e_{k}}P_{\left\{ x\right\} } \\
&&+i\langle \mathfrak{e}_{y+e_{q}},\Delta _{\omega ,\vartheta }\mathfrak{e}%
_{y}\rangle _{\mathfrak{h}}\langle \mathfrak{e}_{x},\Delta _{\omega
,\vartheta }\mathfrak{e}_{x+e_{k}}\rangle _{\mathfrak{h}}\mathrm{e}%
^{-i\gamma h^{(\omega )}}s_{e_{q}}P_{\left\{ y\right\} }\mathrm{e}^{i\gamma
h^{(\omega )}}s_{-e_{k}}P_{\left\{ x+e_{k}\right\} } \\
&&+i\langle \mathfrak{e}_{y},\Delta _{\omega ,\vartheta }\mathfrak{e}%
_{y+e_{q}}\rangle _{\mathfrak{h}}\langle \mathfrak{e}_{x+e_{k}},\Delta
_{\omega ,\vartheta }\mathfrak{e}_{x}\rangle _{\mathfrak{h}}\mathrm{e}%
^{-i\gamma h^{(\omega )}}s_{-e_{q}}P_{\left\{ y+e_{q}\right\} }\mathrm{e}%
^{i\gamma h^{(\omega )}}s_{e_{k}}P_{\left\{ x\right\} } \\
&&-i\langle \mathfrak{e}_{y},\Delta _{\omega ,\vartheta }\mathfrak{e}%
_{y+e_{q}}\rangle _{\mathfrak{h}}\langle \mathfrak{e}_{x},\Delta _{\omega
,\vartheta }\mathfrak{e}_{x+e_{k}}\rangle _{\mathfrak{h}}\mathrm{e}%
^{-i\gamma h^{(\omega )}}s_{-e_{q}}P_{\left\{ y+e_{q}\right\} }\mathrm{e}%
^{i\gamma h^{(\omega )}}s_{-e_{k}}P_{\left\{ x+e_{k}\right\} }.
\end{eqnarray*}%
Using this last equality together with (\ref{shift})-(\ref{shift2}) and (\ref%
{N}), we thus get that%
\begin{eqnarray*}
\mathbf{N}_{\gamma ,q,k}^{(L,\omega )}\mathfrak{e}_{z} &=&\underset{%
x,y,x+e_{k},y+e_{q}\in \Lambda _{L}}{\sum } \\
&&\left\{ i\langle \mathfrak{e}_{x+e_{k}},\Delta _{\omega ,\vartheta }%
\mathfrak{e}_{x}\rangle _{\mathfrak{h}}\langle \mathfrak{e}_{y+e_{q}},\Delta
_{\omega ,\vartheta }\mathfrak{e}_{y}\rangle _{\mathfrak{h}}\langle 
\mathfrak{e}_{x},\mathrm{e}^{-i\gamma h^{(\omega )}}\mathfrak{e}%
_{y+e_{q}}\rangle _{\mathfrak{h}}\langle \mathfrak{e}_{y},\mathrm{e}%
^{i\gamma h^{(\omega )}}\mathfrak{e}_{z}\rangle _{\mathfrak{h}}\mathfrak{e}%
_{x+e_{k}}\right. \\
&&-i\langle \mathfrak{e}_{x+e_{k}},\Delta _{\omega ,\vartheta }\mathfrak{e}%
_{x}\rangle _{\mathfrak{h}}\langle \mathfrak{e}_{y},\Delta _{\omega
,\vartheta }\mathfrak{e}_{y+e_{q}}\rangle _{\mathfrak{h}}\langle \mathfrak{e}%
_{x},\mathrm{e}^{-i\gamma h^{(\omega )}}\mathfrak{e}_{y}\rangle _{\mathfrak{h%
}}\langle \mathfrak{e}_{y+e_{q}},\mathrm{e}^{i\gamma h^{(\omega )}}\mathfrak{%
e}_{z}\rangle _{\mathfrak{h}}\mathfrak{e}_{x+e_{k}} \\
&&-i\langle \mathfrak{e}_{x},\Delta _{\omega ,\vartheta }\mathfrak{e}%
_{x+e_{k}}\rangle _{\mathfrak{h}}\langle \mathfrak{e}_{y+e_{q}},\Delta
_{\omega ,\vartheta }\mathfrak{e}_{y}\rangle _{\mathfrak{h}}\langle 
\mathfrak{e}_{x+e_{k}},\mathrm{e}^{-i\gamma h^{(\omega )}}\mathfrak{e}%
_{y+e_{q}}\rangle _{\mathfrak{h}}\langle \mathfrak{e}_{y},\mathrm{e}%
^{i\gamma h^{(\omega )}}\mathfrak{e}_{z}\rangle _{\mathfrak{h}}\mathfrak{e}%
_{x} \\
&&+i\langle \mathfrak{e}_{x},\Delta _{\omega ,\vartheta }\mathfrak{e}%
_{x+e_{k}}\rangle _{\mathfrak{h}}\langle \mathfrak{e}_{y},\Delta _{\omega
,\vartheta }\mathfrak{e}_{y+e_{q}}\rangle _{\mathfrak{h}}\langle \mathfrak{e}%
_{x+e_{k}},\mathrm{e}^{-i\gamma h^{(\omega )}}\mathfrak{e}_{y}\rangle _{%
\mathfrak{h}}\langle \mathfrak{e}_{y+e_{q}},\mathrm{e}^{i\gamma h^{(\omega
)}}\mathfrak{e}_{z}\rangle _{\mathfrak{h}}\mathfrak{e}_{x} \\
&&-i\delta _{x,z}\langle \mathfrak{e}_{y+e_{q}},\Delta _{\omega ,\vartheta }%
\mathfrak{e}_{y}\rangle _{\mathfrak{h}}\langle \mathfrak{e}_{x+e_{k}},\Delta
_{\omega ,\vartheta }\mathfrak{e}_{x}\rangle _{\mathfrak{h}}\langle 
\mathfrak{e}_{y},\mathrm{e}^{i\gamma h^{(\omega )}}\mathfrak{e}%
_{x+e_{k}}\rangle _{\mathfrak{h}}\mathrm{e}^{-i\gamma h^{(\omega )}}%
\mathfrak{e}_{y+e_{q}} \\
&&+i\delta _{x+e_{k},z}\langle \mathfrak{e}_{y+e_{q}},\Delta _{\omega
,\vartheta }\mathfrak{e}_{y}\rangle _{\mathfrak{h}}\langle \mathfrak{e}%
_{x},\Delta _{\omega ,\vartheta }\mathfrak{e}_{x+e_{k}}\rangle _{\mathfrak{h}%
}\langle \mathfrak{e}_{y},\mathrm{e}^{i\gamma h^{(\omega )}}\mathfrak{e}%
_{x}\rangle _{\mathfrak{h}}\mathrm{e}^{-i\gamma h^{(\omega )}}\mathfrak{e}%
_{y+e_{q}} \\
&&+i\delta _{x,z}\langle \mathfrak{e}_{y},\Delta _{\omega ,\vartheta }%
\mathfrak{e}_{y+e_{q}}\rangle _{\mathfrak{h}}\langle \mathfrak{e}%
_{x+e_{k}},\Delta _{\omega ,\vartheta }\mathfrak{e}_{x}\rangle _{\mathfrak{h}%
}\langle \mathfrak{e}_{y+e_{q}},\mathrm{e}^{i\gamma h^{(\omega )}}\mathfrak{e%
}_{x+e_{k}}\rangle _{\mathfrak{h}}\mathrm{e}^{-i\gamma h^{(\omega )}}%
\mathfrak{e}_{y} \\
&&\left. -i\delta _{x+e_{k},z}\langle \mathfrak{e}_{y},\Delta _{\omega
,\vartheta }\mathfrak{e}_{y+e_{q}}\rangle _{\mathfrak{h}}\langle \mathfrak{e}%
_{x},\Delta _{\omega ,\vartheta }\mathfrak{e}_{x+e_{k}}\rangle _{\mathfrak{h}%
}\langle \mathfrak{e}_{y+e_{q}},\mathrm{e}^{i\gamma h^{(\omega )}}\mathfrak{e%
}_{x}\rangle _{\mathfrak{h}}\mathrm{e}^{-i\gamma h^{(\omega )}}\mathfrak{e}%
_{y}\right\} .
\end{eqnarray*}%
By using (\ref{equation sup}) and (\ref{Combes-ThomasCombes-Thomas})-(\ref%
{Combes-ThomasCombes-Thomasbis}) together with%
\begin{equation*}
\sum_{z\in \mathbb{Z}^{d}}\mathrm{e}^{-2\mu _{\eta }\left(
|x-z|+|y-z|\right) }\leq \mathrm{e}^{-\mu _{\eta }|x-y|}\sum_{z\in \mathbb{Z}%
^{d}}\mathrm{e}^{-\mu _{\eta }\left( |x-z|+|y-z|\right) }\leq \mathrm{e}%
^{-\mu _{\eta }|x-y|}\sum_{z\in \mathbb{Z}^{d}}\mathrm{e}^{-2\mu _{\eta
}|z|}.
\end{equation*}%
(which are simple consequences of Cauchy-Schwarz and triangle inequalities),
all the above summands are absolutely summable, uniformly with respect to $%
L\in \mathbb{R}^{+}$, $\omega \in \Omega $, $\lambda \in \mathbb{R}_{0}^{+}$
and $\vartheta $,$\gamma $ in compact subsets of $\mathbb{R}_{0}^{+}$ and $%
\mathbb{R}$, respectively. For instance, for any (characteristic) functions $%
f,g:\mathbb{Z}^{d}\rightarrow \left\{ 0,1\right\} $, one estimates that%
\begin{multline*}
\underset{x,y\in \mathbb{Z}^{d}}{\sum }f\left( x\right) ^{2}g\left( y\right)
^{2}\left\vert \langle \mathfrak{e}_{x+e_{k}},\Delta _{\omega ,\vartheta }%
\mathfrak{e}_{x}\rangle _{\mathfrak{h}}\langle \mathfrak{e}_{y+e_{q}},\Delta
_{\omega ,\vartheta }\mathfrak{e}_{y}\rangle _{\mathfrak{h}}\langle 
\mathfrak{e}_{x},\mathrm{e}^{-i\gamma h^{(\omega )}}\mathfrak{e}%
_{y+e_{q}}\rangle _{\mathfrak{h}}\langle \mathfrak{e}_{y},\mathrm{e}%
^{i\gamma h^{(\omega )}}\mathfrak{e}_{z}\rangle _{\mathfrak{h}}\right\vert
\left\Vert \mathfrak{e}_{x+e_{k}}\right\Vert _{\mathfrak{h}} \\
\leq 36^{2}\left( 1+\vartheta \right) ^{2}\mathrm{e}^{2\left\vert \gamma
\eta \right\vert }\underset{x,y\in \mathbb{Z}^{d}}{\sum }f\left( x\right)
^{2}g\left( y\right) ^{2}\mathrm{e}^{-2\mu _{\eta }\left(
|x-e_{q}-y|+|z-y|\right) } \\
\leq 36^{2}\left( 1+\vartheta \right) ^{2}\mathrm{e}^{2\left\vert \gamma
\eta \right\vert }\left( \sum_{u\in \mathbb{Z}^{d}}g\left( u+z\right) ^{2}%
\mathrm{e}^{-2\mu _{\eta }|u|}\right) ^{1/2} \\
\times \underset{x\in \mathbb{Z}^{d}}{\sum }f\left( x\right) ^{2}\mathrm{e}%
^{-\mu _{\eta }|x-e_{q}-z|}\left( \sum_{y\in \mathbb{Z}^{d}}g\left(
y+x-e_{q}\right) ^{2}\mathrm{e}^{-2\mu _{\eta }|y|}\right) ^{1/2}<\infty .
\end{multline*}%
(Recall that $\mu _{\eta }>0$, by (\ref{Combes-ThomasCombes-Thomasbis}).) In
fact, by the same arguments combined with%
\begin{equation*}
\left\Vert C\right\Vert _{\mathcal{B}(\mathfrak{h})}\leq \sup_{x\in \mathbb{Z%
}^{d}}\sum_{z\in \mathbb{Z}^{d}}\left\vert \left\langle \mathfrak{e}_{x},C%
\mathfrak{e}_{z}\right\rangle _{\mathfrak{h}}\right\vert ,\qquad C\in 
\mathcal{B}(\mathfrak{h}),
\end{equation*}%
(see \cite[Lemma 4.1]{LDP}), the absolutely summable sum 
\begin{equation}
\mathrm{e}^{-i\gamma h^{(\omega )}}\mathfrak{e}_{w}=\sum_{u\in \mathbb{Z}%
^{d}}\mathfrak{e}_{u}\langle \mathfrak{e}_{u},\mathrm{e}^{-i\gamma
h^{(\omega )}}\mathfrak{e}_{w}\rangle _{\mathfrak{h}},\qquad w\in \mathbb{Z}%
^{d},  \label{bn}
\end{equation}%
(see (\ref{Combes-ThomasCombes-Thomas})-(\ref{Combes-ThomasCombes-Thomasbis}%
)) and Lebesgue's dominated convergence theorem, in the limit $L\rightarrow
\infty $ and for any $z\in \Lambda _{L/2}$, there is an operator $\mathbf{R}%
_{\gamma ,q,k}^{(L,\omega )}\in \mathcal{B}\left( \mathfrak{h}\right) $ with
vanishing operator norm as $L\rightarrow \infty $, uniformly with respect to 
$\omega \in \Omega $, $\lambda \in \mathbb{R}_{0}^{+}$ and $\vartheta $,$%
\gamma $ in compact subsets of $\mathbb{R}_{0}^{+}$ and $\mathbb{R}$,
respectively, such that 
\begin{equation*}
\mathbf{N}_{\gamma ,q,k}^{(L,\omega )}\mathfrak{e}_{z}=\left( \mathbf{N}%
_{\gamma ,q,k}^{(\infty ,\omega )}+\mathbf{R}_{\gamma ,q,k}^{(L,\omega
)}\right) \mathfrak{e}_{z},
\end{equation*}%
where 
\begin{eqnarray*}
\mathbf{N}_{\gamma ,q,k}^{(\infty ,\omega )}\mathfrak{e}_{z} &\doteq &%
\underset{x,y\in \mathbb{Z}^{d}}{\sum }\left\{ i\langle \mathfrak{e}%
_{x+e_{k}},\Delta _{\omega ,\vartheta }\mathfrak{e}_{x}\rangle _{\mathfrak{h}%
}\langle \mathfrak{e}_{y+e_{q}},\Delta _{\omega ,\vartheta }\mathfrak{e}%
_{y}\rangle _{\mathfrak{h}}\langle \mathfrak{e}_{x},\mathrm{e}^{-i\gamma
h^{(\omega )}}\mathfrak{e}_{y+e_{q}}\rangle _{\mathfrak{h}}\langle \mathfrak{%
e}_{y},\mathrm{e}^{i\gamma h^{(\omega )}}\mathfrak{e}_{z}\rangle _{\mathfrak{%
h}}\mathfrak{e}_{x+e_{k}}\right. \\
&&-i\langle \mathfrak{e}_{x+e_{k}},\Delta _{\omega ,\vartheta }\mathfrak{e}%
_{x}\rangle _{\mathfrak{h}}\langle \mathfrak{e}_{y},\Delta _{\omega
,\vartheta }\mathfrak{e}_{y+e_{q}}\rangle _{\mathfrak{h}}\langle \mathfrak{e}%
_{x},\mathrm{e}^{-i\gamma h^{(\omega )}}\mathfrak{e}_{y}\rangle _{\mathfrak{h%
}}\langle \mathfrak{e}_{y+e_{q}},\mathrm{e}^{i\gamma h^{(\omega )}}\mathfrak{%
e}_{z}\rangle _{\mathfrak{h}}\mathfrak{e}_{x+e_{k}} \\
&&-i\langle \mathfrak{e}_{x},\Delta _{\omega ,\vartheta }\mathfrak{e}%
_{x+e_{k}}\rangle _{\mathfrak{h}}\langle \mathfrak{e}_{y+e_{q}},\Delta
_{\omega ,\vartheta }\mathfrak{e}_{y}\rangle _{\mathfrak{h}}\langle 
\mathfrak{e}_{x+e_{k}},\mathrm{e}^{-i\gamma h^{(\omega )}}\mathfrak{e}%
_{y+e_{q}}\rangle _{\mathfrak{h}}\langle \mathfrak{e}_{y},\mathrm{e}%
^{i\gamma h^{(\omega )}}\mathfrak{e}_{z}\rangle _{\mathfrak{h}}\mathfrak{e}%
_{x} \\
&&+i\langle \mathfrak{e}_{x},\Delta _{\omega ,\vartheta }\mathfrak{e}%
_{x+e_{k}}\rangle _{\mathfrak{h}}\langle \mathfrak{e}_{y},\Delta _{\omega
,\vartheta }\mathfrak{e}_{y+e_{q}}\rangle _{\mathfrak{h}}\langle \mathfrak{e}%
_{x+e_{k}},\mathrm{e}^{-i\gamma h^{(\omega )}}\mathfrak{e}_{y}\rangle _{%
\mathfrak{h}}\langle \mathfrak{e}_{y+e_{q}},\mathrm{e}^{i\gamma h^{(\omega
)}}\mathfrak{e}_{z}\rangle _{\mathfrak{h}}\mathfrak{e}_{x} \\
&&-i\delta _{x,z}\langle \mathfrak{e}_{y+e_{q}},\Delta _{\omega ,\vartheta }%
\mathfrak{e}_{y}\rangle _{\mathfrak{h}}\langle \mathfrak{e}_{x+e_{k}},\Delta
_{\omega ,\vartheta }\mathfrak{e}_{x}\rangle _{\mathfrak{h}}\langle 
\mathfrak{e}_{y},\mathrm{e}^{i\gamma h^{(\omega )}}\mathfrak{e}%
_{x+e_{k}}\rangle _{\mathfrak{h}}\mathrm{e}^{-i\gamma h^{(\omega )}}%
\mathfrak{e}_{y+e_{q}} \\
&&+i\delta _{x+e_{k},z}\langle \mathfrak{e}_{y+e_{q}},\Delta _{\omega
,\vartheta }\mathfrak{e}_{y}\rangle _{\mathfrak{h}}\langle \mathfrak{e}%
_{x},\Delta _{\omega ,\vartheta }\mathfrak{e}_{x+e_{k}}\rangle _{\mathfrak{h}%
}\langle \mathfrak{e}_{y},\mathrm{e}^{i\gamma h^{(\omega )}}\mathfrak{e}%
_{x}\rangle _{\mathfrak{h}}\mathrm{e}^{-i\gamma h^{(\omega )}}\mathfrak{e}%
_{y+e_{q}} \\
&&+i\delta _{x,z}\langle \mathfrak{e}_{y},\Delta _{\omega ,\vartheta }%
\mathfrak{e}_{y+e_{q}}\rangle _{\mathfrak{h}}\langle \mathfrak{e}%
_{x+e_{k}},\Delta _{\omega ,\vartheta }\mathfrak{e}_{x}\rangle _{\mathfrak{h}%
}\langle \mathfrak{e}_{y+e_{q}},\mathrm{e}^{i\gamma h^{(\omega )}}\mathfrak{e%
}_{x+e_{k}}\rangle _{\mathfrak{h}}\mathrm{e}^{-i\gamma h^{(\omega )}}%
\mathfrak{e}_{y} \\
&&\left. -i\delta _{x+e_{k},z}\langle \mathfrak{e}_{y},\Delta _{\omega
,\vartheta }\mathfrak{e}_{y+e_{q}}\rangle _{\mathfrak{h}}\langle \mathfrak{e}%
_{x},\Delta _{\omega ,\vartheta }\mathfrak{e}_{x+e_{k}}\rangle _{\mathfrak{h}%
}\langle \mathfrak{e}_{y+e_{q}},\mathrm{e}^{i\gamma h^{(\omega )}}\mathfrak{e%
}_{x}\rangle _{\mathfrak{h}}\mathrm{e}^{-i\gamma h^{(\omega )}}\mathfrak{e}%
_{y}\right\} .
\end{eqnarray*}%
It suffices now to use again (\ref{equation sup}) and (\ref{bn}) together
with elementary manipulations in each sum of $\mathbf{N}_{\gamma
,q,k}^{(\infty ,\omega )}$ in order to arrive at the second assertion.
\end{proof}

We are now in a position to show (\ref{blabla_2blabla_2}), at least for $%
\left\vert \gamma \right\vert ,\vartheta \ll 1$, as a consequence of the
next two lemmata:\ 

\begin{lemma}[Asymptotics for $\protect\vartheta \ll 1$]
\label{les termes en eta et eta^2 sont nuls copy(4)}\mbox{}\newline
\label{long_lemme copy(4)}For all $k,q\in \{1,\ldots ,d\}$, $\vartheta
,\lambda \in \mathbb{R}_{0}^{+}$, $\omega \in \Omega $, $\gamma \in \mathbb{R%
}$ and $z\in \mathbb{Z}^{d}$,%
\begin{equation*}
\sum_{y\in \mathbb{Z}^{d}}\zeta _{z,y,z}=2\Im \mathrm{m}\left\langle \left(
s_{e_{k}}-s_{-e_{k}}\right) \mathfrak{e}_{z},\mathrm{e}^{-i\gamma h^{(\omega
)}}\left( s_{e_{q}}-s_{-e_{q}}\right) \mathrm{e}^{i\gamma h^{(\omega )}}%
\mathfrak{e}_{z}\right\rangle _{\mathfrak{h}}+\mathcal{O}\left( \vartheta
\right) ,\qquad \text{as }\vartheta \rightarrow 0,
\end{equation*}%
uniformly with respect to $\omega \in \Omega $, $\lambda \in \mathbb{R}%
_{0}^{+}$ and $\gamma $ in compact subsets of $\mathbb{R}$. Note that $%
\vartheta $ is not necessarily $0$ in definition of $h^{(\omega )}$.
\end{lemma}

\begin{proof}
By Lemma \ref{long_lemme copy(3)} at $\vartheta =0$, one directly computes
that, for any $k,q\in \{1,\ldots ,d\}$, $\lambda \in \mathbb{R}_{0}^{+}$, $%
\omega \in \Omega $, $\gamma \in \mathbb{R}$, $z\in \mathbb{Z}^{d}$ and $%
\vartheta =0$, 
\begin{equation*}
\sum_{y\in \mathbb{Z}^{d}}\zeta _{z,y,z}=\underset{y\in \mathbb{Z}^{d}}{\sum 
}2\Im \mathrm{m}\langle \mathfrak{e}_{z+e_{k}}-\mathfrak{e}_{z-e_{k}},%
\mathrm{e}^{-i\gamma h^{(\omega )}}\left( \mathfrak{e}_{y+e_{q}}-\mathfrak{e}%
_{y-e_{q}}\right) \rangle _{\mathfrak{h}}\langle \mathfrak{e}_{y},\mathrm{e}%
^{i\gamma h^{(\omega )}}\mathfrak{e}_{z}\rangle _{\mathfrak{h}}.
\end{equation*}%
If $\vartheta \neq 0$ then one performs the same kind of computation in
order to (trivially) deduce the assertion, by (\ref{shift}), Lemma \ref%
{long_lemme copy(3)} and (\ref{Combes-ThomasCombes-Thomas})-(\ref%
{Combes-ThomasCombes-Thomasbis}).
\end{proof}

\begin{lemma}[Asymptotics for $\left\vert \protect\gamma \right\vert \ll 1$]

\label{long_lemme copy(5)}\mbox{}\newline
For all $k,q\in \{1,\ldots ,d\}$, $\vartheta ,\lambda \in \mathbb{R}_{0}^{+}$%
, $\omega \in \Omega $, $\gamma \in \mathbb{R}$ and $z\in \mathbb{Z}^{d}$, 
\begin{multline*}
2\Im \mathrm{m}\left\langle \left( s_{e_{k}}-s_{-e_{k}}\right) \mathfrak{e}%
_{z},\mathrm{e}^{-i\gamma h^{(\omega )}}\left( s_{e_{q}}-s_{-e_{q}}\right) 
\mathrm{e}^{i\gamma h^{(\omega )}}\mathfrak{e}_{z}\right\rangle _{\mathfrak{h%
}} \\
=2\gamma \lambda \delta _{k,q}\left\{ 2\omega _{1}\left( z\right) -\omega
_{1}\left( z+e_{k}\right) -\omega _{1}\left( z-e_{k}\right) \right\} +%
\mathcal{O}\left( \gamma ^{2}\right) ,
\end{multline*}%
as $\left\vert \gamma \right\vert \rightarrow 0$, uniformly with respect to $%
\omega \in \Omega $ and $\vartheta ,\lambda $ in compact subsets of $\mathbb{%
R}_{0}^{+}$.
\end{lemma}

\begin{proof}
By (\ref{borne trivial}), for any $\gamma \in \mathbb{R}$,%
\begin{equation*}
\mathrm{e}^{i\gamma h^{(\omega )}}=\mathbf{1}_{\mathfrak{h}}+\sum_{n\in 
\mathbb{N}}\frac{\left( i\gamma h^{(\omega )}\right) ^{n}}{n!}=\mathbf{1}_{%
\mathfrak{h}}+i\gamma h^{(\omega )}+\mathcal{O}\left( \gamma ^{2}\right)
,\qquad \text{as }\left\vert \gamma \right\vert \rightarrow 0,
\end{equation*}%
in the\ Banach space $\mathcal{B}\left( \mathfrak{h}\right) $, uniformly
with respect to $\omega \in \Omega $ and $\vartheta ,\lambda $ in compact
subsets of $\mathbb{R}_{0}^{+}$. The assertion then follows by direct
computations using (\ref{equation sup})-(\ref{eq:Ham_lap_pot}), (\ref{shift}%
) and the last equality.
\end{proof}

\begin{lemma}[Lower bounds on the Hilbert-Schmidt norm of $K_{\{\Lambda
_{L}\},\{\mathbb{Z}^{d}\}}^{(\protect\omega ,\mathcal{E})}$]
\label{long_lemme copy(6)}\label{les termes en eta et eta^2 sont nuls
copy(3)}\mbox{}\newline
Take $\vartheta ,\lambda ,T\in \mathbb{R}_{0}^{+}$, $T\in \mathbb{R}^{+}$, $%
\mathcal{E}\in C_{0}^{0}(\mathbb{R};\mathbb{R}^{d})$ with support in $[-T,0]$
and $\vec{w}\doteq (w_{1},\ldots ,w_{d})\in {\mathbb{R}}^{d}$ with $%
\left\Vert \vec{w}\right\Vert _{\mathbb{R}^{d}}=1$. If $T,\vartheta $ are
sufficiently small then 
\begin{eqnarray*}
&&\lim_{L\rightarrow \infty }\frac{1}{\left\vert \Lambda _{L}\right\vert }%
\mathbb{E}\left[ \mathrm{Tr}_{\mathfrak{h}}\left( \left( K_{\{\Lambda
_{L}\},\{\mathbb{Z}^{d}\}}^{(\cdot ,\mathcal{E})}\right) ^{\ast
}K_{\{\Lambda _{L}\},\{\mathbb{Z}^{d}\}}^{(\cdot ,\mathcal{E})}\right) %
\right] \\
&\geq &\frac{\lambda ^{2}}{2}\mathrm{Var}\left[ \int_{-\infty
}^{0}\left\langle w^{(\cdot )},\mathcal{E}\left( \alpha \right)
\right\rangle _{\mathbb{R}^{d}}\alpha ^{2}\mathrm{d}\alpha \right] +\mathcal{%
O}\left( \vartheta ^{2}\right) +\mathcal{O}\left( T^{4}\right) ,
\end{eqnarray*}%
uniformly with respect to $\lambda $ in compact subsets of $\mathbb{R}%
_{0}^{+}$, where $w^{(\cdot )}\doteq (w_{1}^{(\cdot )},\ldots ,w_{d}^{(\cdot
)})\in \mathbb{R}^{d}$ is the random vector defined by 
\begin{equation}
w_{k}^{(\omega )}\doteq \left( 2\omega _{1}\left( 0\right) -\omega
_{1}\left( e_{k}\right) -\omega _{1}\left( -e_{k}\right) \right)
w_{k},\qquad k\in \{1,\ldots ,d\},\text{ }\omega \in \Omega .
\label{inequalty fluctuation 0}
\end{equation}
\end{lemma}

\begin{proof}
Fix all parameters of the lemma. Take any $L\geq 2$. Note that%
\begin{equation}
\mathrm{Tr}_{\mathfrak{h}}\left( \left( K_{\{\Lambda _{L}\},\{\mathbb{Z}%
^{d}\}}^{(\omega ,\mathcal{E})}\right) ^{\ast }K_{\{\Lambda _{L}\},\{\mathbb{%
Z}^{d}\}}^{(\omega ,\mathcal{E})}\right) \geq \sum_{z\in \Lambda
_{L/2}}\left\Vert K_{\{\Lambda _{L}\},\{\mathbb{Z}^{d}\}}^{(\omega ,\mathcal{%
E})}\mathfrak{e}_{z}\right\Vert _{\mathfrak{h}}^{2}\geq \sum_{z\in \Lambda
_{L/2}}\left\vert \left\langle \mathfrak{e}_{z},K_{\{\Lambda _{L}\},\{%
\mathbb{Z}^{d}\}}^{(\omega ,\mathcal{E})}\mathfrak{e}_{z}\right\rangle _{%
\mathfrak{h}}\right\vert ^{2}.  \label{ddddddddd}
\end{equation}%
By using (\ref{definition Kbisdefinition Kbis})-(\ref{N}) and Lemma \ref%
{long_lemme copy(3)}, for any $z\in \Lambda _{L/2}$, we have that 
\begin{eqnarray*}
\left\langle \mathfrak{e}_{z},K_{\{\Lambda _{L}\},\{\mathbb{Z}%
^{d}\}}^{(\cdot ,\mathcal{E})}\mathfrak{e}_{z}\right\rangle _{\mathfrak{h}}
&=&\underset{k,q=1}{\sum^{d}}w_{k}\int_{-\infty }^{0}\left\{ \mathcal{E}%
\left( \alpha \right) \right\} _{q}\int\nolimits_{0}^{-\alpha }\sum_{y\in 
\mathbb{Z}^{d}}\zeta _{z,y,z}\ \mathrm{d}\gamma \mathrm{d}\alpha \\
&&+\underset{k,q=1}{\sum^{d}}w_{k}\int_{-\infty }^{0}\left\{ \mathcal{E}%
\left( \alpha \right) \right\} _{q}\int\nolimits_{0}^{-\alpha }\left\langle 
\mathfrak{e}_{z},\mathbf{R}_{\gamma ,q,k}^{(L,\omega )}\mathfrak{e}%
_{z}\right\rangle _{\mathfrak{h}}\mathrm{d}\gamma \mathrm{d}\alpha
\end{eqnarray*}%
with $\mathbf{R}_{\gamma ,q,k}^{(L,\omega )}\in \mathcal{B}\left( \mathfrak{h%
}\right) $ satisfying (\ref{limit}). Note that $\zeta _{z,y,z}$ is $\gamma $%
-dependent and its explicit expression is found in Lemma \ref{long_lemme
copy(3)}. If $T,\vartheta $ are sufficiently small then, by Lemmata \ref{les
termes en eta et eta^2 sont nuls copy(4)}-\ref{long_lemme copy(5)}, we
deduce that, for any $z\in \Lambda _{L/2}$,%
\begin{eqnarray*}
\left\langle \mathfrak{e}_{z},K_{\{\Lambda _{L}\},\{\mathbb{Z}%
^{d}\}}^{(\cdot ,\mathcal{E})}\mathfrak{e}_{z}\right\rangle _{\mathfrak{h}}
&=&\lambda \underset{k=1}{\sum^{d}}w_{k}\int_{-\infty }^{0}\left\{ 2\omega
_{1}\left( z\right) -\omega _{1}\left( z+e_{k}\right) -\omega _{1}\left(
z-e_{k}\right) \right\} \left\{ \mathcal{E}\left( \alpha \right) \right\}
_{k}\alpha ^{2}\mathrm{d}\alpha \\
&&+\mathcal{O}\left( \vartheta \right) +\mathcal{O}\left( T^{2}\right) +%
\underset{k,q=1}{\sum^{d}}w_{k}\int_{-\infty }^{0}\left\{ \mathcal{E}\left(
\alpha \right) \right\} _{q}\int\nolimits_{0}^{-\alpha }\left\langle 
\mathfrak{e}_{z},\mathbf{R}_{\gamma ,q,k}^{(L,\omega )}\mathfrak{e}%
_{z}\right\rangle _{\mathfrak{h}}\mathrm{d}\gamma \mathrm{d}\alpha
\end{eqnarray*}%
uniformly with respect to $\omega \in \Omega $ and $\lambda $ in compact
subsets of $\mathbb{R}_{0}^{+}$. By the translation invariance of the
distribution $\mathfrak{a}_{\Omega }$ (see \cite[Equations (1)-(2)]{LDP})
and (\ref{limit}), it follows that 
\begin{eqnarray*}
\lim_{L\rightarrow \infty }\mathbb{E}\left[ \left\vert \left\langle 
\mathfrak{e}_{z},K_{\{\Lambda _{L}\},\{\mathbb{Z}^{d}\}}^{(\cdot ,\mathcal{E}%
)}\mathfrak{e}_{z}\right\rangle _{\mathfrak{h}}\right\vert ^{2}\right]
&=&\lambda ^{2}\mathbb{E}\left[ \left\vert \int_{-\infty }^{0}\left\langle
w^{(\cdot )},\mathcal{E}\left( \alpha \right) \right\rangle _{\mathbb{R}%
^{d}}\alpha ^{2}\mathrm{d}\alpha \right\vert ^{2}\right] +\mathcal{O}\left(
\vartheta ^{2}\right) +\mathcal{O}\left( T^{4}\right) \\
&=&\lambda ^{2}\mathrm{Var}\left[ \int_{-\infty }^{0}\left\langle w^{(\cdot
)},\mathcal{E}\left( \alpha \right) \right\rangle _{\mathbb{R}^{d}}\alpha
^{2}\mathrm{d}\alpha \right] +\mathcal{O}\left( \vartheta ^{2}\right) +%
\mathcal{O}\left( T^{4}\right) ,
\end{eqnarray*}%
uniformly with respect to $\lambda $ in compact subsets of $\mathbb{R}%
_{0}^{+}$. Thanks to (\ref{ddddddddd}), the assertion then follows. Note that%
\begin{equation*}
\mathbb{E}\left[ \int_{-\infty }^{0}\left\langle w^{(\cdot )},\mathcal{E}%
\left( \alpha \right) \right\rangle _{\mathbb{R}^{d}}\alpha ^{2}\mathrm{d}%
\alpha ^{2}\right] =0.
\end{equation*}
\end{proof}

By combining Lemmata \ref{long_lemme copy(2)}, \ref{les termes en eta et
eta^2 sont nuls copy(3)} and \ref{equation inmportantelemma}, we directly
obtain that, for any $\vartheta ,\lambda ,T\in \mathbb{R}_{0}^{+}$, $T,\beta
\in \mathbb{R}^{+}$, $\mathcal{E}\in C_{0}^{0}(\mathbb{R};\mathbb{R}^{d})$
with support in $[-T,0]$ and $\vec{w}\in {\ \mathbb{R}}^{d}$ with $%
\left\Vert \vec{w}\right\Vert _{\mathbb{R}^{d}}=1$,%
\begin{equation}
\partial _{s}^{2}\mathrm{J}^{(s\mathcal{E})}|_{s=0}\geq \frac{1}{2\left( 1+%
\mathrm{e}^{\beta \left( 2d\left( 2+\vartheta \right) +\lambda \right)
}\right) ^{2}}\left( \lambda ^{2}\mathrm{Var}\left[ \int_{-\infty
}^{0}\left\langle w^{(\cdot )},\mathcal{E}\left( \alpha \right)
\right\rangle _{\mathbb{R}^{d}}\alpha ^{2}\mathrm{d}\alpha \right] +\mathcal{%
O}\left( \vartheta ^{2}\right) +\mathcal{O}\left( T^{4}\right) \right) ,
\label{inequalty fluctuation 1}
\end{equation}%
provided that $T,\vartheta $ are sufficiently small. In particular, if 
\begin{equation}
\mathrm{Var}\left[ \int_{-\infty }^{0}\left\langle w^{(\cdot )},\mathcal{E}%
\left( \alpha \right) \right\rangle _{\mathbb{R}^{d}}\alpha ^{2}\mathrm{d}%
\alpha \right] >0  \label{inequalty fluctuation 2}
\end{equation}%
then $\partial _{s}^{2}\mathrm{J}^{(s\mathcal{E})}|_{s=0}>0$. This last
condition is easily satisfied:\ Because the variance of the sum (or the
difference) of uncorrelated random variables is the sum of their variances,
if the random variables $\omega _{1}\left( 0\right) ,\omega _{1}\left(
e_{1}\right) ,\omega _{1}\left( -e_{1}\right) ,\ldots ,\omega _{1}\left(
e_{d}\right) ,\omega _{1}\left( -e_{d}\right) $ are independently and
identically distributed (i.i.d.), then 
\begin{eqnarray}
\mathbb{E}\left[ \left\vert \int_{-\infty }^{0}\left\langle w^{(\omega )},%
\mathcal{E}\left( \alpha \right) \right\rangle _{\mathbb{R}^{d}}\alpha ^{2}%
\mathrm{d}\alpha \right\vert ^{2}\right] &=&2\mathrm{Var}\left[ (\cdot
)_{1}\left( 0\right) \right] \times \left( 2\left( \int_{-\infty
}^{0}\left\langle w,\mathcal{E}\left( \alpha \right) \right\rangle _{\mathbb{%
R}^{d}}\alpha ^{2}\mathrm{d}\alpha \right) ^{2}\right.  \notag \\
&&\left. +\sum_{k=1}^{d}\left( w_{k}\int_{-\infty }^{0}\left( \mathcal{E}%
\left( \alpha \right) \right) _{k}\alpha ^{2}\mathrm{d}\alpha \right)
^{2}\right) ,  \label{iid}
\end{eqnarray}%
which is strictly positive as soon as $\mathcal{E}\neq 0$ and $\omega
_{1}\left( 0\right) $ is not almost surely constant, by Chebychev's
inequality. \bigskip

\noindent \textit{Acknowledgments:} \textbf{\ }This work is supported by
CNPq (308337/2017-4), as well as by the Basque Government through the grant
IT641-13 and the BERC 2018-2021 program and by the Spanish Ministry of
Science, Innovation and Universities: BCAM Severo Ochoa accreditation
SEV-2017-0718, MTM2017-82160-C2-2-P. We thank S. Rodrigues for having
pointed out \cite{bio}.

\end{document}